\documentclass[journal]{IEEEtran}
\usepackage{ifpdf}
 \ifpdf
 \else
 \fi

\usepackage{cite}

\ifCLASSINFOpdf
 \usepackage[pdftex]{graphicx}

\else

 \usepackage[dvips]{graphicx}

\fi
\usepackage{overpic}
\usepackage{amsmath,amssymb,amsfonts,latexsym}
\usepackage{upgreek}

\newtheorem{remark}{Remark}
\newtheorem{proposition}{Proposition}

\newenvironment{proof}{{\hspace{0.3cm}\it Proof:}\hspace{0.15cm}}{\par}
\allowdisplaybreaks[4]

\usepackage{algorithm,float}
\usepackage{algorithmic}
\usepackage{lipsum}


\usepackage{bbding}

\usepackage{array}

\usepackage{fixltx2e}

\usepackage{color}

\usepackage{hyperref}

\usepackage{enumitem}

\usepackage{comment}

\usepackage{bm}

\usepackage{booktabs}

\usepackage{stfloats}

\hypersetup{
	pdftitle={Draft},
	pdfauthor={Junshan Luo},
	bookmarks=true,
	bookmarksopen=true,
	bookmarksnumbered=true,
	colorlinks=true,
	linkcolor=blue,
	citecolor=green,
	urlcolor=magenta
}

\begin{document}
%
\title{Large-scale Aerial Reconfigurable Intelligent Surface-aided Robust Anti-jamming Transmission}

\author{Junshan~Luo, Shilian~Wang, and~Boxiang He
	\thanks{The authors are with the College of Electronic Science and Technology, National University of Defense Technology, Changsha 410005, China (e-mail: luojunshan10@nudt.edu.cn, wangsl@nudt.edu.cn, boxianghe1@bjtu.edu.cn)}
	\thanks{Junshan Luo is also with the Sixty-Third Research Institute, National University of Defense Technology, Nanjing 210007, China.}
}
\maketitle

\begin{abstract}
Aerial reconfigurable intelligent surfaces (ARIS), deployed on unmanned aerial vehicles (UAVs), could enhance anti-jamming communication performance by dynamically configuring channel conditions and establishing reliable air-ground links. 
However, large-scale ARIS faces critical deployment challenges due to the prohibitive computational complexity of conventional discrete optimization methods and sophisticated jamming threats. 
In this paper, we introduce a mean field modeling approach to design the spatial configuration of ARIS by a continuous density function, thus bypassing high-dimensional combinatorial optimization.
We consider an adaptive jammer which adjusts its position and beamforming to minimize the sum-rate.
A key finding reveals that the jammer's optimal strategy is governed by a proximity-directivity trade-off between reducing path loss and enhancing spatial focusing.
To combat the jamming, we propose a robust anti-jamming transmission framework that jointly optimizes the BS beamforming, the ARIS reflection, and the ARIS spatial distribution to maximize the worst-case sum-rate. 
By leveraging variational optimization and Riemannian manifold methods, we efficiently solve the functional optimization problems.
Our analysis further unveils that the optimal ARIS deployment follows a spatial water-filling principle, concentrating resources in high-gain regions while avoiding interference-prone areas. 
Simulation results demonstrate that the proposed framework remarkably improves the sum-rate.
Furthermore, the computational complexity of the proposed algorithm is independent of the number of UAVs, validating its effectiveness for scalable ARIS-assisted anti-jamming communications.
\end{abstract}


\section{Introduction}
Ensuring reliable wireless communications in adversarial environments poses a critical challenge to modern information networks \cite{Pirayesh_H2022}. Jamming attacks, intentionally designed to disrupt legitimate transmissions, can severely compromise system throughput, latency, and security \cite{Darsena_D2022}. As such, robust anti-jamming techniques are indispensable to safeguard next-generation wireless systems against increasingly sophisticated threats.

In this context, reconfigurable intelligent surface (RIS) has emerged as a revolutionary technology to proactively defend against jamming threats by intelligently reshaping the wireless environment \cite{Renzo_DM2020,Wu_Q2021}. 
Composed of a large number of passive meta-atoms, an RIS can dynamically manipulate the phase, amplitude, and polarization of incident electromagnetic waves, thereby constructing favorable signal pathways for legitimate users and creating nulls towards jammers without the need for expensive radio frequency chains or high power consumption \cite{Huang_C2020}.
Moving a step further, deploying RIS on unmanned aerial vehicles (UAVs) to form aerial RIS (ARIS) unlocks even greater potential. 
ARIS combines the channel reconfiguration capability of RIS with the high mobility and superior line-of-sight (LoS) conditions of UAVs, enabling on-demand and enhanced coverage in certain areas \cite{Mahmoud_A2021}. 
This makes ARIS particularly promising for rapid deployment in temporary hotspots, emergency communications, and adversarial scenarios where jamming is prevalent.

Recognizing this potential, several studies have explored the application of RIS and ARIS for anti-jamming communications.
Early works showed that optimizing the transmit power and RIS reflection under unknown, adaptive jammers can substantially improve transmission resilience, where fast reinforcement learning was used to track jammer dynamics \cite{Yang_H2021}.
Beyond single-surface settings, the work in \cite{Zhang_M2024,Dong_H2024,ElMossallamy_2021} further demonstrated that RIS can shape nulls toward jammers and steer constructive paths to users across various network models such as wireless powered networks \cite{Zhang_M2024}, integrated terrestrial-satellite networks \cite{Dong_H2024}, device-to-device (D2D) networks \cite{ElMossallamy_2021}, etc.
To cope with strong adversaries, beamforming against jamming has also been developed for advanced RIS variants such as multi-layer RIS \cite{Sun_Y2024,Zou_C2023}, simultaneous transmitting and reflecting RIS \cite{Zhou_T2024}, showing improved anti-jamming gains versus conventional RIS.
Going airborne, ARIS leverages UAV mobility to place the reflecting array at positions that optimally mitigate the impact of jamming signals, coupling passive beamforming with placement or trajectory design.
The work in \cite{Tang_X2021} formulated joint passive beamforming and ARIS deployment, and verified significant jamming suppression. 
Subsequent works have extended ARIS to dynamic multi-cell multi-user scenarios, where joint optimization of BS beamforming, RIS phase shifts, user scheduling, and UAV trajectory leads to improved minimum ergodic user rates \cite{Liu_J2024}. 
In maritime contexts, the authors of \cite{Yang_H2024_1} introduced an energy-harvesting ARIS system using deep reinforcement learning to jointly optimize transmit power, ARIS placement, and RIS beamforming, which demonstrated enhanced energy efficiency and endurance under jamming threats in maritime channels.
In D2D networks, \cite{Hou_Z2023} proposed a matching-based joint ARIS selection and passive beamforming mechanism across multiple ARIS units, improving anti-jamming performance via transmitter-ARIS matching in anti-jamming D2D communications.
Moreover, recent advances have explored UAV swarm enabled ARIS architectures, which overcome the limitations of single-UAV payload and enhance aperture gain, spatial multiplexing, and robustness through cooperative deployment and beamforming design \cite{Shang_B2023}.

However, the extension of these approaches to large-scale ARIS networks, which comprise swarms of RIS-mounted drones, remains particularly challenging due to the prohibitive computational complexity of optimizing a massive number of variables including UAV positions and RIS phase shifts.
To tackle the inherent complexity of optimizing large-scale drone swarms, mean field theory (MFT) offers a powerful mathematical framework by transforming the intractable discrete optimization over numerous individual agents into a tractable continuous problem via a spatial density function \cite{Wang_D2024,Jiang_J2025}. 
Instead of optimizing the precise location of each UAV, MFT characterizes the collective deployment of the ARIS system through a density distribution dramatically reducing the algorithmic complexity and providing crucial analytical insights. This methodology has found successful applications in large-scale wireless networks, such as power control in multiple access systems \cite{Benamor_A2022} and energy-efficient velocity control in massive UAV networks \cite{Gao_H2022}, where it demonstrates superior scalability and efficiency. 

Despite the promising advances of ARIS and MFT, several research gaps exist.
First, most existing studies on ARIS-assisted anti-jamming communications have been limited to small-scale scenarios, with little emphasis on the scalability and resilience of massive ARIS networks under adversarial jamming conditions. Second, while MFT has demonstrated effectiveness in resource allocation tasks such as power control and trajectory optimization, its potential for the joint optimization of UAV spatial deployment and RIS phase shifts in adversarial environments has yet to be fully explored. Third, as adaptive jammers are highly dynamic and their positions may be uncertain or deliberately concealed, which significantly complicates interference management in ARIS-assisted networks. Although recent works have incorporated uncertainty in the jamming channel, they largely neglect a rigorous treatment of the optimal jammer location problem.
To the best of our knowledge, a joint optimization framework that is robust to jammer location uncertainty in the context of MFT-based optimization for large-scale ARIS systems has not been investigated.

To bridge these gaps, this paper proposes a MFT-based framework for robust anti-jamming transmission in large-scale ARIS-assisted systems. The main contributions of this work are summarized as follows:
\begin{itemize}
\item \textit{Mean field modeling for ARIS deployment:} We propose the mean-field approach to model the spatial configuration of a large-scale ARIS system using a continuous density function, which decouples the computational complexity from the ARIS swarm size, enabling efficient optimization for massive ARIS networks.
\item \textit{Robust worst-case formulation against an adaptive jammer:} We consider an adaptive jammer that optimizes its transmission strategy, including both beamforming and position selection, to minimize the system sum-rate. A key insight from this formulation is that the proximity-directivity trade-off governs the jammer's optimal strategy. Our formulation jointly optimizes the base station (BS) beamforming, ARIS reflection coefficients, and ARIS spatial density to maximize the worst-case sum-rate under this jammer uncertainty.
\item \textit{An efficient optimization framework:} We solve the formulated problem based on variational optimization and Riemannian manifold. Specifically, the non-convex unit-modulus constraints for the RIS phases are rigorously handled via Riemannian gradient ascent on a complex circle manifold.
The variational optimization on ARIS spatial deployment shows to follow a spatial water-filling principle. In particular, we propose a density-tuned aerial resource allocation (DT-ARA) algorithm that allocates ARISs to locations where their net marginal gain exceeds a threshold, analogous to power allocation over parallel channels.
\item \textit{Comprehensive performance validation:} Through extensive simulations, we demonstrate that the proposed framework remarkably improves the sum-rate and jamming resilience compared to state-of-the-art benchmarks. The results validate the effectiveness of our approach for scalable and robust ARIS-assisted communications in contested environments.
\end{itemize}

The remainder of this paper is organized as follows. Section \ref{sec:system_model} introduces the system model and problem formulation. Section \ref{sec:jamming_strategy} analyzes the adaptive jammer's strategy. The proposed joint optimization algorithms are detailed in Section \ref{sec:4}. Complexity and convergence analyses are provided in Section \ref{sec:performance}. Numerical results are presented in Section \ref{sec:numerical_results}, followed by conclusions in Section \ref{sec:conclusion}. 

\begin{table}[t]
	\centering
	\caption{Summary of Key Mathematical Notations}
	\label{tab:notations}
	\begin{tabular}{ll}
		\toprule
		\textbf{Parameter} & \textbf{Definition} \\
		\midrule
		$M$ & Number of BS antennas \\
		$N$ & Number of RIS elements per UAV \\
		$Q$ & Total number of UAVs \\
		$K$ & Number of single-antenna users \\
		$L$ & Number of jammer antennas \\
		$\mathcal{D}$ & UAV deployment region, $\mathcal{D} \subset \mathbb{R}^2$ \\
		$\rho(\mathbf{x})$ & UAV spatial density at $\mathbf{x} \in \mathcal{D}$ \\
		$\rho_{\max}$ & Max UAV density constraint \\
		$\hat{\mathbf{j}}$ & Estimated jammer location \\
		$\epsilon$ & Jammer position uncertainty \\
		$\mathcal{J}$ & Jammer uncertain region \\
		$\mathbf{H}_{\textrm{B,U}}(\mathbf{x})$ & BS-ARIS channel, $\mathbf{H}_{\textrm{B,U}}(\mathbf{x}) \in \mathbb{C}^{N \times M}$ \\
		$\mathbf{h}_{\textrm{U},k}(\mathbf{x})$ & ARIS-$k$th user channel, $\mathbf{h}_{\textrm{U},k}(\mathbf{x}) \in \mathbb{C}^{N \times 1}$ \\
		$\mathbf{H}_{\textrm{J,U}}(\mathbf{x}, \mathbf{j})$ & Jammer-ARIS channel, $\mathbf{H}_{\textrm{J,U}}(\mathbf{x}, \mathbf{j}) \in \mathbb{C}^{N \times L}$ \\
		$\bm{\Theta}(\mathbf{x})$ & RIS reflection matrix at $\mathbf{x}$\\
		$\theta_n(\mathbf{x})$ & Phase shift of the $n$th element at $\mathbf{x}$\\
		$\mathbf{v}$ & Jammer beamforming \\
		$\mathbf{w}_k$ & BS beamforming for user $k$ \\
		$P_\text{B}$ & Max BS transmit power \\
		$P_\text{J}$ & Jamming power \\
		$\sigma_k^2$ & Noise at user $k$ \\
		$\mathbf{h}_{\text{eff, B},k}$ & Aggregate BS-$k$th user channel, c.f. \eqref{eq:agg_BS_user} \\
		$h_{\text{eff,J},k}$ & Aggregate jammer-$k$th user channel, c.f. \eqref{eq:agg_jam_user} \\
		$\gamma_k$ & SINR, c.f. \eqref{eq:SINR} \\
		$R_{\text{sum}}$ & Sum-rate, c.f. \eqref{eq:sum_rate} \\
		$\mathbf{b}_k(\mathbf{j})$ & Component of $h_{\text{eff,J},k}$ \\
		$\mathbf{R}(\mathbf{j})$ & Jamming covariance, $\sum_k \mathbf{b}_k(\mathbf{j}) \mathbf{b}_k^\dagger(\mathbf{j})$ \\
		$\lambda_{\max}(\mathbf{R}(\mathbf{j}))$ & Max eigenvalue of $\mathbf{R}(\mathbf{j})$ \\
		$z_{\text{B},k}$ & Signal gain of $k$th user, c.f. \eqref{eq:z_Bk} \\
		$z_{\text{J},k}$ & Jamming gain, c.f. \eqref{eq:z_Jk} \\
		$c_k$ & Interference term, c.f. \eqref{eq:c_k} \\
		$B_k$ & Signal enhancement in gradient, c.f. \eqref{eq:B_k}\\
		$C_k$ & Jamming suppression in gradient, c.f. \eqref{eq:C_k} \\
		$\mathcal{M}$ & RIS phase manifold, c.f. \eqref{eq:RIS_manifold} \\
		$G(\mathbf{x})$ & Net marginal gain function, c.f. \eqref{eq:net_gain_funtion} \\
		$ \eta(\mathbf{x}) $ & An arbitrary test function \\
		\bottomrule
	\end{tabular}
\end{table}

\textit{Notations}: $a$, $\mathbf{a}$, $\mathbf{A}$, and $\mathcal{A}$ denote a scalar, a vector, a matrix, and a set, respectively. $(\cdot)^{\rm{T}}$, $(\cdot)^{\dagger}$, and $(\cdot)^{-1}$ denote transpose, conjugate transpose, and inverse, respectively.
$[x]^+$ denotes the positive part of a real number $x$, defined as $[x]^+ = \mathrm{max}\{ 0, x \}$.
For a vector $\mathbf{a}$, $[\mathbf{a}]_n$ denotes the $n$th element. $\mathcal{CN}(0,\sigma)$ denotes the circularly symmetric complex Gaussian (CSCG) distribution with mean zero and covariance $\sigma$. $\mathbb{R}$ and $\mathbb{C}$ represent the sets of real and complex numbers, respectively. $\mathrm{Re}(\cdot)$, $\mathrm{Im}(\cdot)$, and $|\cdot|$ denote the real part, the imaginary part, and the amplitude of a complex number or complex vector, respectively. $\partial(\cdot)$ denotes the partial differential of a function or the boundary of a set, depending on the context. $\delta(\cdot)$ denotes the functional derivative. Define $\mathcal{I}_N = \left\{1, 2, \ldots, N \right\}$ as a shorthand as the index set. For a set $\mathcal{A}$, $\mathrm{int}(\mathcal{A})$ denotes its interior. 
The key variables of this paper are listed in Table \ref{tab:notations}.

\section{System Model and Problem Formulation} \label{sec:system_model}

As shown in Fig. \ref{fig:system_model}, we consider a downlink communication system comprising of a multi-antenna BS, multiple UAVs, multiple single-antenna users, and a multi-antenna jammer. 
The BS has $M$ antennas and is located at the ground level.
$Q$ UAVs are deployed at fixed altitude, each carrying a RIS with $N$ passive reflecting elements, forming an ARIS swarm.
Their horizontal positions $\mathbf{x} \in \mathcal{D}$, where $\mathcal{D} \subset \mathbb{R}^2$ denotes the feasible region.
The BS is serving $K$ single-antenna terrestrial users via the ARIS swarm.
The malicious jammer with $L$ antennas is located at an uncertain horizontal position $\mathbf{j} \in \mathcal{J}$, where $\mathcal{J} = \{ \mathbf{j} \in \mathbb{R}^2:  \| \mathbf{j} - \hat{\mathbf{j}} \| \leq \epsilon \}$; $\hat{\mathbf{j}}$ is the roughly estimated location and $\epsilon$ denotes the jammer position uncertainty.

\subsection{Mean Field Effective Channels}
For the ARIS at $\mathbf{x}$, the BS-ARIS channel is denoted by $\mathbf{H}_{\textrm{B,U}}(\mathbf{x}) \in \mathbb{C}^{N \times M}$.
For user $k$ and ARIS at $\mathbf{x}$, the ARIS-user channel is denoted by $\mathbf{h}_{\textrm{U},k}(\mathbf{x}) \in \mathbb{C}^{N \times 1}$, where $k \in \mathcal{I}_K$.
The channel between jammer at $\mathbf{j}$ and ARIS at $\mathbf{x}$ is denoted by $\mathbf{H}_{\textrm{J,U}}(\mathbf{x}, \mathbf{j}) \in \mathbb{C}^{N \times L}$.\footnote{Building on the air-ground propagation characteristics, the jammer-ARIS channels can be accurately modeled using LoS-dominant formulations, particularly at millimeter-wave (mmWave) and higher frequency bands. Given the LoS-dominant nature, coarse jammer position estimation $\hat{\mathbf{j}}$ enables sufficient jamming channel reconstruction. The exclusion of BS-user and jammer-user direct links is justified by the following two reasons. First, at mmWave or higher frequencies, ground-level links experience severe diffraction loss and are highly susceptible to blockage by terrestrial infrastructures. Second, the ARIS swarm creates preferential air-ground paths with near-omnidirectional coverage. Therefore, both the BS and the jammer prioritize directing their beams toward the ARISs.}
The ARIS at $\mathbf{x}$ has a reflection matrix $\boldsymbol{\Theta}(\mathbf{x}) = \mathrm{diag} \left\{ \theta_1(\mathbf{x}), \dots, \theta_N(\mathbf{x}) \right\}$,
satisfying the unit-modulus constraint $\left| \theta_n(\mathbf{x}) \right| = 1$, $n \in \mathcal{I}_N$.
Based on the MFT, we describe the spatial distribution of ARIS with continuous density $\rho(\mathbf{x})$, which denotes the density of UAVs at region $\mathbf{x}$,
satisfying $\int_{\mathcal{D}} { \rho(\mathbf{x}) } d\mathbf{x} = Q $ with maximum density constraint to avoid collision, i.e., $\rho(\mathbf{x}) \leq \rho_{\max}$.
Specifically, the aggregate BS to the $k$th user channel via all ARISs, denoted by $\mathbf{h}_{\text{eff,B},k}^{\dagger} \in \mathbb{C}^{1 \times M}$, is expressed as
\begin{equation}\label{eq:agg_BS_user}
\mathbf{h}_{\text{eff,B},k}^\dagger = \int_{\mathcal{D}} \mathbf{h}_{\text{U},k}^\dagger(\mathbf{x}) \boldsymbol{\Theta}(\mathbf{x}) \mathbf{H}_{\text{B,U}}(\mathbf{x}) \rho(\mathbf{x}) \mathrm{d}\mathbf{x}.
\end{equation}

The aggregate jamming channel for the $k$th user $h_{\text{eff,J},k}(\mathbf{j},\mathbf{v}) \in \mathbb{C}$ is given by
\begin{equation}\label{eq:agg_jam_user}
h_{\text{eff,J},k}(\mathbf{j},\mathbf{v}) = \int_{\mathcal{D}} \mathbf{h}_{\text{U},k}^\dagger (\mathbf{x}) \boldsymbol{\Theta}(\mathbf{x}) \mathbf{H}_{\text{J,U}}(\mathbf{x},\mathbf{j}) \mathbf{v} \rho(\mathbf{x}) \mathrm{d} \mathbf{x} 
\end{equation}
where $\mathbf{v} \in \mathbb{C}^{L \times 1}$ is the  beamforming vector of the jammer, satisfying $\| \mathbf{v} \| = 1$.

\subsection{Signal Model and Sum-rate}
\begin{figure}[t] 
	\centering
	\includegraphics[width=0.4\textwidth]{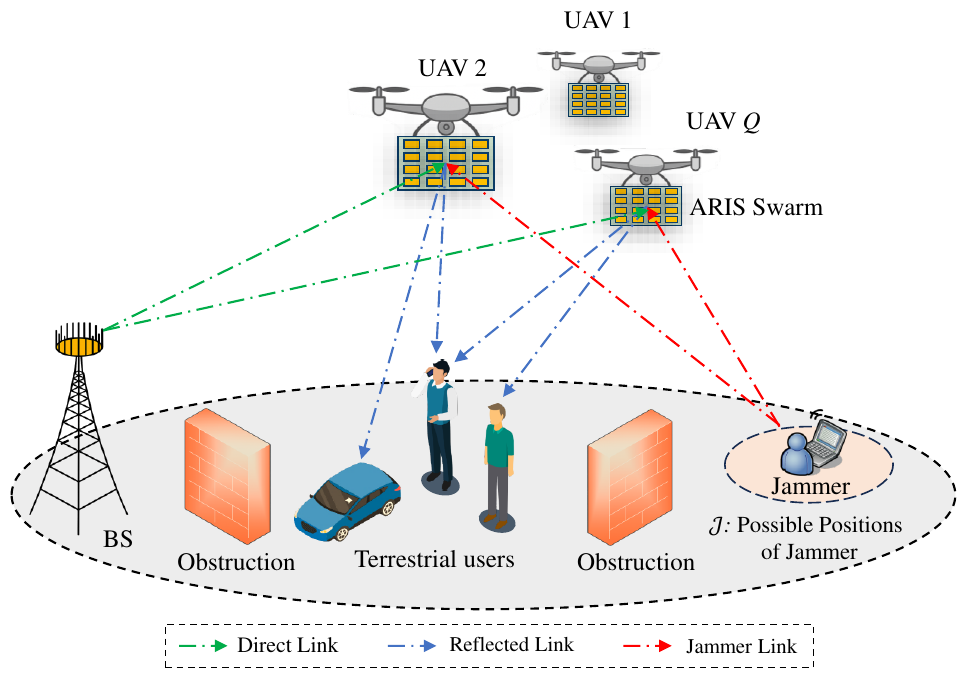} 
	\caption{An illustration of ARIS swarm-aided anti-jamming communications system.} 
	\label{fig:system_model}
\end{figure}
The received signal at user $k$ is 
\begin{equation}\label{eq:y_k}
y_k = \mathbf{h}_{\text{eff,B},k}^\dagger \sum_{k = 1}^{K} \mathbf{w}_k s_{\text{U}, k} + h_{\text{eff,J},k}(\mathbf{j},\mathbf{v}) s_{\text{J}} + n_k
\end{equation}
where $\mathbf{w}_k \in \mathbb{C}^{M \times 1}$ denotes the BS beamforming vector for user $k$; $s_{\text{U},k} \sim \mathcal{CN}(0, 1)$ is the information symbol for user $k$; 
$s_{\text{J}} \sim \mathcal{CN}(0, P_{\text{J}})$ is the jamming signal with power $P_{\text{J}}$; $n_k \sim \mathcal{CN}(0, \sigma_k^2)$ is the noise with power $\sigma_k^2$.
From \eqref{eq:y_k}, the SINR of user $k$ is given by
\begin{equation}\label{eq:SINR}
\gamma_k = \frac{\left| \mathbf{h}_{\text{eff,B},k}^\dagger \mathbf{w}_k \right|^2}{\sum\limits_{i=1, i \neq k}^{K} \left| \mathbf{h}_{\text{eff,B},k}^\dagger \mathbf{w}_i \right|^2 + \left| h_{\text{eff,J},k}(\mathbf{j},\mathbf{v}) \right|^2 P_\text{J} + \sigma_k^2}
\end{equation}
and the sum-rate is given by
\begin{equation}\label{eq:sum_rate}
R_{\textrm{sum}} = \sum_{k = 1}^{K} \log_2(1 + \gamma_k).
\end{equation}

\subsection{Problem Formulation}
The goal is to maximize the sum-rate by jointly optimizing the spatial distribution of the ARIS swarm, the BS beamforming vectors, and the RIS reflection coefficients, under practical constraints. The problem is formulated as\footnote{The interaction between the communication system and the adaptive jammer would ideally be formulated as a max-min optimization problem, i.e., $\max_{\rho, \{ \mathbf{w}_k \}, \boldsymbol{\Theta}} \min_{\mathbf{j} \in \mathcal{J}, \| \mathbf{v} \| = 1} R_{\text{sum}}$. However, in practice, the jammer's strategy adaptation occurs on a significantly longer time scale compared to the nearly instantaneous adjustment of the BS beamforming and RIS phases. Leveraging this temporal disparity, we adopt a pragmatic approach where the worst-case jamming strategy is computed a priori for a given system configuration and is treated as a quasi-static parameter during the optimization of the communication system. The practical operation issues are further detailed in remark \ref{remark:jammer_strategy}.}
\begin{subequations}
	\label{eq:opt_1}
	\begin{align}
    & \max_{\rho, \{ \mathbf{w}_k \}, \boldsymbol{\Theta}} R_{\text{sum}} \\
	& \qquad \text{s.t. } \sum_{k = 1}^{K} \| \mathbf{w}_k \|^2 \leq P_\text{B}  \label{cons:BS_power} \\
	& \qquad \phantom{\text{s.t. }} \left| \theta_n(\mathbf{x}) \right| = 1, n \in \mathcal{I}_N \label{cons:unit_modulus} \\
	& \qquad \phantom{\text{s.t. }} \mathbf{j} \in \mathcal{J}, \quad \| \mathbf{v} \| = 1 \label{cons:jammer_pos} \\
	& \qquad \phantom{\text{s.t. }} \int_{\mathcal{D}} { \rho(\mathbf{x}) } \mathrm{d} \mathbf{x} = Q, 0 \leq \rho(\mathbf{x}) \leq \rho_{\max} \label{cons:density}
	\end{align}
\end{subequations}
where constraint \eqref{cons:BS_power} imposes a total power budget $P_\text{B}$ at the BS; \eqref{cons:unit_modulus} enforces the unit-modulus condition on each RIS reflecting element, which arises from the passive nature of RIS hardware and implies that each element can only adjust the phase of the incident signal without amplifying its amplitude; constraint \eqref{cons:jammer_pos} accounts for the uncertainty in the jammer's position and the jammer's beamforming vector is also constrained to be unit-norm; \eqref{cons:density} ensures that the spatial density function $\rho(\mathbf{x})$ integrates to the total number of UAVs $Q$ over the deployment region $\mathcal{D}$, with an upper bound $\rho_{\max}$ to prevent excessive UAV density and avoid collisions.

The key feature of \eqref{eq:opt_1} is that UAV positions are optimized via a continuous density function instead of discrete coordinates, which avoids the complexity scaling with the number of UAVs $Q$.
The challenges are three-fold. First, the optimization variables include a continuous spatial density function $\rho(\mathbf{x})$ and continuous RIS phase functions $\theta_n(\mathbf{x})$, which are infinite-dimensional in their problem formulation and lead to a functional optimization problem that is computationally demanding. Second, the unit-modulus constraint \eqref{cons:unit_modulus} and the jammer uncertainty in \eqref{cons:jammer_pos} are inherently non-convex. Third, the problem requires robustness against an adaptive jammer that optimally chooses its position and beamforming to minimize the sum-rate, which results in a nested minimax optimization that is challenging to solve.
Overall, problem \eqref{eq:opt_1} is a large-scale, non-convex, functional optimization problem under uncertainty.

\section{Jamming Strategy Analysis}\label{sec:jamming_strategy}
In this section, we analyze the optimal strategy of an adaptive jammer that minimizes the sum-rate by jointly optimizing its position and beamforming vector, under the uncertainty region constraint. The derived strategy characterizes the worst-case jamming attack for a fixed system configuration, forming the basis for our robust transmission design.

\subsection{Jammer's Optimal Beamforming}
Since $R_{\textrm{sum}}$ is monotonically decreasing with the interference term, minimizing $R_{\textrm{sum}}$ is equivalent to maximizing the jamming power. 
The mathematical formulation is 
\begin{equation}\label{jamming_power_max}
\max_{\mathbf{j} \in \mathcal{J}, \| \mathbf{v} \| = 1} \sum_{k = 1}^{K} \left| h_{\text{eff,J},k}(\mathbf{j}, \mathbf{v}) \right|^2.
\end{equation}

Let $\mathbf{b}_k(\mathbf{j}) \triangleq \int_{\mathcal{D}} \mathbf{H}_{\text{J,U}}^\dagger(\mathbf{x}, \mathbf{j}) \boldsymbol{\Theta}^\dagger(\mathbf{x}) \mathbf{h}_{\text{U},k}(\mathbf{x}) \rho(\mathbf{x}) \mathrm{d}\mathbf{x} \in \mathbb{C}^{L \times 1}$.
The objective function of \eqref{jamming_power_max} can be rewritten as 
\begin{align}
\sum_{k = 1}^{K} \left| h_{\text{eff,J},k}(\mathbf{j}, \mathbf{v}) \right|^2 &= \sum_{k = 1}^{K} \left| \mathbf{v}^\dagger \mathbf{b}_k(\mathbf{j}) \right|^2 \\
&= \mathbf{v}^\dagger \underbrace{\left( \sum_{k = 1}^{K} \mathbf{b}_k(\mathbf{j}) \mathbf{b}_k^\dagger(\mathbf{j}) \right)}_{\mathbf{R}(\mathbf{j})} \mathbf{v}
\end{align}
where $\mathbf{R}(\mathbf{j}) \triangleq \sum_{k = 1}^{K} \mathbf{b}_k(\mathbf{j}) \mathbf{b}_k^\dagger(\mathbf{j}) $ represents the jamming covariance matrix that captures spatial correlation of interference across all users.

For a fixed jamming position, the problem \eqref{jamming_power_max} becomes
\begin{equation}\label{jammer_bemforming}
\max_{\|\mathbf{v}\| = 1} \mathbf{v}^\dagger \mathbf{R}(\mathbf{j}) \mathbf{v}
\end{equation}
which is a standard Rayleigh quotient problem and the optimal solution $\mathbf{v}_{\text{opt}} $ is the unit-norm eigenvector corresponding to the largest eigenvalue of $\mathbf{R}(\mathbf{j})$.
The physical insight of $\mathbf{v}_{\text{opt}} $ indicates that the jammer's optimal beamforming transmits along the spatial direction that maximizes coherence with the aggregate jamming channels to the users.

\subsection{Jammer's Optimal Location}
Substituting $\mathbf{v}_{\text{opt}} $ into \eqref{jammer_bemforming}, the jammer's location optimizing problem is reformulated as
\begin{equation}\label{jammer_location_prob}
\max_{\mathbf{j} \in \mathcal{J}} \lambda_{\max}(\mathbf{R}(\mathbf{j}))
\end{equation}
where $\lambda_{\max}( \mathbf{R}(\mathbf{j}) )$ is the maximum eigenvalue of $\mathbf{R}(\mathbf{j})$.
We present the following proposition to address the jammer's optimal deployment strategy.

\begin{proposition}\label{prop_jammer_location}
	The optimal jammer position $\mathbf{j}_\text{opt}$ that maximizes the interference power under position uncertainty constraint satisfies the following two conditions. 
	1) Interior critical point condition: if there exists $\mathbf{j}' \in \mathrm{int}( \mathcal{J} )$ such that the spatial gradient of the maximum eigenvalue vanishes, i.e., $\nabla_{\mathbf{j}} \lambda_{\max}(\mathbf{R}(\mathbf{j}')) = \mathbf{0}$, and $\mathbf{j}'$ is a local maximum, then $\mathbf{j}_\text{opt} = \mathbf{j}'$.
	2) Boundary condition: when no such interior critical point exists, the optimum lies on the boundary $\partial \mathcal{J} = \{ \mathbf{j} | \| \mathbf{j} - \hat{\mathbf{j}} \| = \epsilon \}$ with solution
	\begin{equation}\label{boundary}
	\mathbf{j}_\text{opt} = \hat{\mathbf{j}} + \epsilon \frac{\nabla_{\mathbf{j}} \lambda_{\max}(\mathbf{R}(\mathbf{j})) }{\| \nabla_{\mathbf{j}} \lambda_{\max}(\mathbf{R}(\mathbf{j})) \|}.
	\end{equation}
\end{proposition}
\begin{proof}
	See appendix \ref{appendix1}.
\end{proof}

To derive key insights from Proposition \ref{prop_jammer_location}, we present the gradient direction as
\begin{equation}\label{gradient_direction}
\nabla_{\mathbf{j}} \lambda_{\max}(\mathbf{R}(\mathbf{j})) \!\propto\!\! \int_{\mathcal{D}}\! \underbrace{\frac{\alpha}{2} \frac{\mathbf{x} - \mathbf{j}}{\| \mathbf{x} - \mathbf{j} \|^{\alpha + 2}}}_{\text{Distance gradient}} \!\cdot\! \underbrace{\left| \mathbf{v}_{\text{opt}}^\dagger \mathbf{a}_k(\mathbf{x}, \mathbf{j}) \right|^2}_{\text{Beam-aligned gain}} \rho(\mathbf{x}) \mathrm{d}\mathbf{x}
\end{equation}
where $\alpha$ denotes the path-loss exponent; $\mathbf{a}_k(\mathbf{x}, \mathbf{j}) \triangleq \mathbf{H}_{\text{J,U}}^\dagger(\mathbf{x}, \mathbf{j}) \boldsymbol{\Theta}^\dagger(\mathbf{x}) \mathbf{h}_{\text{U},k}(\mathbf{x}) \in \mathbb{C}^{L \times 1}$ represents the position-dependent response vector.
The gradient direction \eqref{gradient_direction} reveals that there exists a proximity-directivity trade-off and the optimal jamming position emerges from this fundamental trade-off.
Specifically, the interior critical point condition represents a balanced spatial configuration where the jamming position achieves optimal trade-off between minimizing path loss to user clusters and maximizing beamforming coherence gain. 
This typically occurs when users are symmetrically distributed around the jammer position.
On the other hand, the boundary condition indicates that when $\nabla_{\mathbf{j}} \lambda_{\max}(\mathbf{R}(\mathbf{j}))$ points toward user clusters, it emphasizes proximity. 
This occurs when users or UAVs are densely clustered near $\hat{\mathbf{j}}$.
The solution prioritizes minimizing average distance to targets, trading off some directivity gain for significant path loss reduction.
When $\nabla_{\mathbf{j}} \lambda_{\max}(\mathbf{R}(\mathbf{j}))$ points away, it reflects that users are spatially dispersed or at large distances.
The solution prioritizes maximizing angular separation for coherent beamforming, trading off path loss for enhanced spatial focusing capability.

In summary, Proposition \ref{prop_jammer_location} provides a rigorous performance lower-bound for robust jamming mitigation in contested environments.
The implementation details for optimal jamming beamforming and positioning are shown in Algorithm \ref{Jammer_opt_alg}.

\begin{algorithm}[!t]
	\caption{Optimal Jammer Beamforming and Positioning}
	\label{Jammer_opt_alg}
	\begin{algorithmic}[1]
		\REQUIRE Channel matrices $\mathbf{H}_{\text{J,U}}(\mathbf{x}, \mathbf{j})$, $\mathbf{h}_{\text{U},k}(\mathbf{x})$, density $\rho(\mathbf{x})$, RIS reflection $\boldsymbol{\Theta}(\mathbf{x})$, estimated jammer location $\hat{\mathbf{j}}$, position uncertainty $\epsilon$
		\ENSURE Optimal $\mathbf{j}_\text{opt}$, $\mathbf{v}_\text{opt}$
		\STATE Compute $\mathbf{b}_k(\mathbf{j})$ and $\mathbf{R}(\mathbf{j})$ at $\mathbf{j} = \hat{\mathbf{j}}$
		\STATE Find critical points: solve $\nabla_{\mathbf{j}} \lambda_{\max}(\mathbf{R}(\mathbf{j})) = \mathbf{0}$
		\IF {$\exists$ critical point $\mathbf{j}' \in \mathcal{J}$ and $\lambda_{\max}(\mathbf{j}') \geq \lambda_{\max}(\hat{\mathbf{j}})$}
		\STATE $\mathbf{j}_{\text{int}} = \mathbf{j}'$
		\ELSE
		\STATE Compute boundary candidate using \eqref{boundary} and obtain $\mathbf{j}_{\text{bd}}$ 
		\ENDIF
		\STATE Compare solutions: $\mathbf{j}_\text{opt} = \underset{\mathbf{j} \in \{ \mathbf{j}_{\text{int}}, \mathbf{j}_{\text{bd}} \}}{\text{arg max}} \lambda_{\max}(\mathbf{R}(\mathbf{j}))$
		\STATE Compute $\mathbf{v}_\text{opt}$ by solving \eqref{jammer_bemforming}.
	\end{algorithmic}
\end{algorithm}

\begin{remark}
	For practical positioning without calculating the gradient, we can use the following heuristic positioning rules.
	For clustered users within the uncertainty region, $\mathbf{j}_\text{opt}$ locates at the centroid of the user cluster.
	For dispersed users, $\mathbf{j}_\text{opt}$ moves to the boundary along the direction of principal component of user positions.
\end{remark}

\section{Joint Communication System Optimization}\label{sec:4}
Following the characterization of the optimal jamming strategy, this section details the optimization of the communication system parameters to maximize the sum-rate. We employ an alternating optimization approach to jointly design the BS beamforming vectors, the RIS reflection coefficients, and the ARIS spatial density distribution.
\subsection{Beamforming Optimization at the BS}
We focus on optimizing the BS beamforming vectors $\{ \mathbf{w}_k \}$ while keeping all other variables fixed.
To alleviate inter-user interference and maximize the sum-rate, we employ a classic two-step approach: zero-forcing (ZF) beamforming to orthogonalize the user channels, followed by optimal power allocation via the water-filling algorithm. This method is well-established in multi-user multi-antenna systems \cite{Zhang_LH2020, Zhang_S2022, Emil_2014} and efficiently handles the BS power constraint \eqref{cons:BS_power}.

First, we define the composite effective channel matrix for the $K$ users as
\begin{equation}
\mathbf{H}_{\text{eff}} = 
\begin{bmatrix}
\mathbf{h}_{\text{eff,B},1}^\dagger \\
\mathbf{h}_{\text{eff,B},2}^\dagger \\
\vdots \\
\mathbf{h}_{\text{eff,B},K}^\dagger
\end{bmatrix} \in \mathbb{C}^{ K \times M }.
\end{equation}

The ZF beamforming matrix $\mathbf{W}_{\text{zf}} \in \mathbb{C}^{M \times K}$ is derived using the pseudoinverse of $\mathbf{H}_{\text{eff}}$, i.e.,
\begin{equation}
\mathbf{W}_{\text{zf}} = \mathbf{H}_{\text{eff}}^\dagger \left( \mathbf{H}_{\text{eff}} \mathbf{H}_{\text{eff}}^\dagger \right)^{-1} \boldsymbol{\Gamma}
\end{equation}
where $\boldsymbol{\Gamma} \in \mathbb{C}^{K \times K}$ is a diagonal matrix of normalization factors to ensure unit gain per user.
Assume that $p_k$ is the power allocated to user $k$. The sum-rate maximization problem now reduces to optimizing $\{ p_k \}$ under the BS power constraint, given by
\begin{subequations}
	\begin{align}
	& \max_{\{p_k\}} \sum_{k = 1}^{K} \log_2 \left( 1 + \frac{p_k}{ \left| h_{\text{eff,J},k}(\mathbf{j},\mathbf{v}) \right|^2 P_\text{J} + \sigma_k^2 } \right) \\
	& \quad \text{s.t. } \sum_{k = 1}^{K} p_k \leq P_\text{B}, \quad p_k \geq 0 
	\end{align}
\end{subequations}
which is a convex optimization problem solved by the water-filling algorithm.
The optimal power allocation is
\begin{equation}
p_k = [ \tilde{\eta} - \left| h_{\text{eff,J},k}(\mathbf{j},\mathbf{v}) \right|^2 P_\text{J} - \sigma_k^2 ]^+
\end{equation}
where $\tilde{\eta}$ is the water level chosen to satisfy the power constraint \eqref{cons:BS_power}.
The final beamforming vector $\mathbf{w}_k$ incorporating both ZF and power allocation is
\begin{equation}\label{eq:ZF+water_filling}
\mathbf{w}_k = \sqrt{p_k} \tilde{\mathbf{w}}_k
\end{equation}
where $\tilde{\mathbf{w}}_k \in \mathbb{C}^{M \times 1}$ represents the $k$th column of $\mathbf{W}_{\text{zf}}$. 

\subsection{Manifold Optimization of RIS Reflection Coefficients}
With fixed ZF beamforming, the inter-user interference is eliminated and the sub-problem w.r.t. RIS reflection is re-formulated as
\begin{subequations}
	\label{opt_RIS}
	\begin{align}
	& \max_{\boldsymbol{\Theta}} \sum_{k = 1}^{K} \log_2 \left(1 + \frac{\left| \mathbf{h}_{\text{eff,B},k}^\dagger \mathbf{w}_k \right|^2}{ \left| h_{\text{eff,J},k}(\mathbf{j},\mathbf{v}) \right|^2 P_\text{J} + \sigma_k^2} \right) \\
	& \quad \text{s.t. } \eqref{cons:unit_modulus}
	\end{align}
\end{subequations}

The RIS reflection optimization is a functional optimization problem over a continuous space.
The key challenges of solving \eqref{opt_RIS} are three-fold. First, both the signal and interference terms depend on $\boldsymbol{\Theta}(\mathbf{x})$ through spatial integrals.
Second, the infinite-dimensional nature of $\theta_n(\mathbf{x})$ prohibits conventional optimization techniques. 
Third, the non-convex unit-modulus constraint induces a complex solution space.
Therefore, we resolve these by formulating the problem on a Riemannian manifold and deriving a closed-form Riemannian gradient for efficient optimization.
Riemannian optimization inherently enforces unit-modulus constraints via projection operators, avoiding feasibility violations common in penalty-based methods \cite{ElMossallamy_2021}.

\subsubsection{Riemannian Manifold Formulation}
The feasible set for $\theta_n(\mathbf{x})$ forms an infinite-dimensional Riemannian manifold
\begin{equation}\label{eq:RIS_manifold}
\mathcal{M} = \left\{ \boldsymbol{\theta} = (\theta_1, \ldots, \theta_N) \mid \theta_n : \mathcal{D} \to \mathbb{C}, \, |\theta_n(\mathbf{x})| = 1 \right\}
\end{equation}
which is an $N$-product of unit circles embedded in a function space.
At $ \boldsymbol{\theta} \in \mathcal{M}$, the tangent space is
\begin{equation}
\mathcal{T}_{\boldsymbol{\theta}} \mathcal{M} = \left\{ \boldsymbol{\xi} = (\xi_1, \ldots, \xi_N) \mid \mathrm{Re} \left( \xi_n(\mathbf{x}) \theta_n^*(\mathbf{x}) \right) = 0 \right\}
\end{equation}
where $\xi_n(\mathbf{x})$ represents admissible perturbations to $\theta_n(\mathbf{x})$.
$\mathrm{Re} \left( \xi_n(\mathbf{x}) \theta_n^*(\mathbf{x}) \right) = 0$ ensures perturbations remain tangent to the unit circle.

\subsubsection{Retraction Map}
The retraction $\mathcal{R}_{\boldsymbol{\theta}} : \mathcal{T}_{\boldsymbol{\theta}} \mathcal{M} \to \mathcal{M}$ projects the updates back to the manifold
\begin{equation}
\mathcal{R}_{\boldsymbol{\theta}}(\boldsymbol{\xi}) = \left( \frac{\theta_n(\mathbf{x}) + \xi_n(\mathbf{x})}{|\theta_n(\mathbf{x}) + \xi_n(\mathbf{x})|} \right)_{n = 1}^N
\end{equation}
which preserves unit modulus.

\subsubsection{Riemannian Gradient Calculation}
The Riemannian gradient $\mathrm{grad}_{\mathcal{M}} R_{\textrm{sum}}$ is derived by projecting the Euclidean gradient $\nabla R_{\textrm{sum}}$ onto $T_{\boldsymbol{\theta}} \mathcal{M}$  
\begin{equation}
\mathrm{grad}_{\mathcal{M}} R_{\textrm{sum}} = \mathrm{Proj}_{\boldsymbol{\theta}}(\nabla R_{\textrm{sum}})
\end{equation}
where the $n$th element at position $\mathbf{x}$ is
\begin{align}
\left[ \mathrm{Proj}_{\boldsymbol{\theta}}(\nabla R_{\textrm{sum}}) \right]_n(\mathbf{x}) &=  [\nabla R_{\textrm{sum}} ]_n(\mathbf{x}) \notag\\
& \hspace{-15mm}- \mathrm{Re} \left( [\nabla R_{\textrm{sum}} ]_n(\mathbf{x}) \theta_n^*(\mathbf{x}) \right) \theta_n(\mathbf{x}).
\end{align}

This projection removes the component of the Euclidean gradient that would violate the unit-modulus constraint.

\begin{proposition}\label{prop_Eu_gradient}
	The Euclidean gradient $\nabla R_{\textnormal{sum}} \in \mathbb{C}^{N \times 1}$ with respect to the RIS reflection coefficients $\boldsymbol{\Theta}(\mathbf{x})$ is given by the functional derivative
	\begin{align}
	\nabla R_{\textnormal{sum}} = \left( \frac{\delta R_{\textnormal{sum}}}{\delta \theta_n(\mathbf{x})} \right)_{n = 1}^N
	\end{align} 
	where the $n$th element at position $\mathbf{x}$ is
	\begin{align}
	\frac{\delta R_{\textnormal{sum}}}{\delta \theta_n(\mathbf{x})} = \rho(\mathbf{x}) \sum_{k = 1}^{K} \frac{1}{(1 + \gamma_k) \ln 2} ( B_k - C_k ) 
	\end{align}
	with the following definitions
	\begin{align}
	B_k &\triangleq \frac{2 \mathrm{Re}\! \left( z_{\text{B},k}^* c_{\text{B},k,n}(\mathbf{x}) \right)}{c_k} \label{eq:B_k} \\
	C_k &\triangleq \frac{2 |z_{\text{B},k}|^2 P_\text{J} \mathrm{Re}\! \left( z_{\text{J},k}^* c_{\text{J},k,n}(\mathbf{x}) \right)}{c_k^2} \label{eq:C_k}\\
	z_{\text{B},k} &\triangleq \mathbf{h}_{\text{eff,B},k}^\dagger \mathbf{w}_k \label{eq:z_Bk}\\
	z_{\text{J},k} &\triangleq h_{\text{eff,J},k}(\mathbf{j}, \mathbf{v}) \label{eq:z_Jk}\\
	c_k &\triangleq |z_{\text{J},k}|^2 P_\text{J} + \sigma_k^2 \label{eq:c_k}\\
	c_{\text{B},k,n}(\mathbf{x}) &\triangleq \left[ \mathbf{h}_{\text{U},k}(\mathbf{x}) \right]_n^* \left[ \mathbf{H}_{\text{B,U}}(\mathbf{x}) \mathbf{w}_k \right]_n \\
	c_{\text{J},k,n}(\mathbf{x}) &\triangleq \left[ \mathbf{h}_{\text{U},k}(\mathbf{x}) \right]_n^* \left[ \mathbf{H}_{\text{J,U}}(\mathbf{x}, \mathbf{j}) \mathbf{v} \right]_n.
	\end{align}
\end{proposition}
\begin{proof}
	See appendix \ref{appendix2}.
\end{proof}

From Proposition \ref{prop_Eu_gradient}, we find that the Euclidean gradient contains two critical components, i.e., $B_k$ and $C_k$.
$B_k$ denotes signal enhancement term from coherent superposition at user $k$, while $C_k$ represents interference suppression term mitigating jammer-induced degradation.
In addition, the spatial density $\rho(\mathbf{x})$ scales gradient magnitude by drone density, emphasizing high-density regions.

The proposed Riemannian manifold optimization framework provides a rigorous mathematical foundation for optimizing RIS reflection coefficients in continuous space, which explicitly handles the non-convex unit-modulus constraints and couples signal and interference physics through closed-form gradients.
The implementation details for optimal RIS phase configuration are shown in Algorithm \ref{alg:ris-manifold-opt}.

\begin{algorithm}[!t]
	\caption{Riemannian Manifold Optimization for Continuous RIS Phase Control}
	\label{alg:ris-manifold-opt}
	\begin{algorithmic}[1]
		\REQUIRE 
		Spatial density $\rho(\mathbf{x})$, channel coefficients for all users and RIS elements, step size, backtracking factor, convergence tolerance, maximum number of iterations.
		
		\ENSURE Optimized phase functions $\theta_n(\mathbf{x})$.
		
		\STATE \textbf{Initialize:} Generate random unit-modulus phase functions. Set iteration counter to zero.
		
		\WHILE{stopping criterion not met \AND within maximum iterations}
		
		\STATE Compute aggregate effective channels for all users via spatial integration.
		
		\STATE Calculate the Euclidean gradient of the sum-rate $\nabla R_{\textrm{sum}}$ using the result of Proposition \ref{prop_Eu_gradient}.
		
		\STATE Project the Euclidean gradient onto the tangent space of the complex circle manifold.
		
		\STATE Update the phase functions by moving along the tangent direction and retracting onto the manifold.
		
		\STATE Adjust the step size via backtracking to ensure objective function improvement.
		
		\STATE Check convergence based on the norm of the Riemannian gradient.
		\STATE Increment the iteration counter.
		\ENDWHILE
		
		\STATE \textbf{Return} The optimized phase functions $\theta_n(\mathbf{x})$.
	\end{algorithmic}
\end{algorithm}

\subsection{Variational Optimization of UAV Spatial Density}
Recall the definitions of the effective signal gain $z_{\text{B},k}$ and the effective jamming interference $z_{\text{J},k}$ from \eqref{eq:z_Bk} and \eqref{eq:z_Jk}, which are functionals of the spatial density $\rho(\mathbf{x})$.
The UAV spatial density optimization problem is thus written as
\begin{subequations}
	\label{opt_density}
	\begin{align}
	\max_{\rho} &\sum_{k = 1}^{K} \log_2 \left(1 + \frac{| z_{\text{B},k} |^2}{| z_{\text{J},k} |^2 P_\text{J} + \sigma_k^2} \right) \\
	\qquad &\text{s.t. } \eqref{cons:density}
	\end{align}
\end{subequations}
which is a constrained functional optimization problem.

We propose a variational framework using calculus of variations \cite{Brunt_2004} to solve this problem.
The Lagrangian functional is
\begin{equation}\label{eq:lagrangian}
\mathcal{L}(\rho) = R_{\text{sum}} + \tau \left( Q - \int_{\mathcal{D}} \rho(\mathbf{x}) d\mathbf{x} \right)
\end{equation}
where $\tau$ is a Lagrange multiplier for the UAV quantity constraint. 
The functional derivative of $\mathcal{L}(\rho)$ with respect to $\rho(\mathbf{x})$ must vanish at optimality, subject to the box constraint \eqref{cons:density}.
In the following, we provide the derivation of functional derivative.

\subsubsection{Variation of Effective Channels}
We introduce a perturbation $\rho(\mathbf{x}) \to \rho(\mathbf{x}) + \zeta \eta(\mathbf{x})$, where $\zeta$ is infinitesimal and $ \eta(\mathbf{x}) $ is an arbitrary test function. The first-order variations in $z_{\text{B},k}$ and $z_{\text{J},k}$ is
\begin{align}
\delta z_{\text{B},k} &= \zeta \int_{\mathcal{D}} f_k(\mathbf{x}) \eta(\mathbf{x}) \mathrm{d}\mathbf{x} \label{eq:delta_s_k}\\
\delta z_{\text{J},k} &= \zeta \int_{\mathcal{D}} g_k(\mathbf{x}) \eta(\mathbf{x}) \mathrm{d}\mathbf{x}
\end{align}
where $f_k(\mathbf{x}) \triangleq \mathbf{h}_{\text{U},k}^\dagger(\mathbf{x}) \boldsymbol{\Theta}(\mathbf{x}) \mathbf{H}_{\text{B,U}}(\mathbf{x}) \mathbf{w}_k$ and $g_k(\mathbf{x}) \triangleq \mathbf{h}_{\text{U},k}^\dagger(\mathbf{x}) \boldsymbol{\Theta}(\mathbf{x}) \mathbf{H}_{\text{J,U}}(\mathbf{x}, \mathbf{j}) \mathbf{v}$ are the position-dependent signal and jamming response functions. 

\subsubsection{Variation of SINR}
The variation in $\gamma_k$ is derived via chain rule
\begin{equation}
\delta \gamma_k = \frac{\partial \gamma_k}{\partial z_{\text{B},k}} \delta z_{\text{B},k} + \frac{\partial \gamma_k}{\partial z_{\text{J},k}} \delta z_{\text{J},k}.
\end{equation}

Recalling the definition of $\gamma_k$ in \eqref{opt_density}, we have
\begin{align}
\frac{\partial \gamma_k}{\partial z_{\text{B},k}} &= \frac{2 \mathrm{Re}(z_{\text{B},k}^* \delta z_{\text{B},k})}{c_k} \\
\frac{\partial \gamma_k}{\partial z_{\text{J},k}} &= - \frac{2 |z_{\text{B},k}|^2 P_\text{J} \mathrm{Re}(z_{\text{J},k}^* \delta z_{\text{J},k})}{c_k^2}. \label{eq:partial_gamma_k}
\end{align}

\subsubsection{Variation of Sum-rate}
The variation in $R_{\text{sum}}$ is 
\begin{equation}
\delta R_{\text{sum}} = \sum_{k = 1}^{K} \frac{1}{\ln 2} \frac{1}{1 + \gamma_k} \delta \gamma_k. \label{eq:delta_R_sum}
\end{equation}

Substituting \eqref{eq:delta_s_k}-\eqref{eq:partial_gamma_k} into \eqref{eq:delta_R_sum} yields
\begin{equation}
\delta R_{\textrm{sum}} =  \zeta \int_{\mathcal{D}} G(\mathbf{x})  \eta(\mathbf{x}) d\mathbf{x}
\end{equation}
where $G(\mathbf{x})$ denotes the net marginal gain function with expression
\begin{equation}\label{eq:net_gain_funtion}
G(\mathbf{x}) \triangleq \sum_{k = 1}^{K} \frac{2}{\ln 2} \frac{ \mathrm{Re}(z_{\text{B},k}^* f_k(\mathbf{x})) - P_\text{J} \gamma_k \mathrm{Re}(z_{\text{J},k}^* g_k(\mathbf{x})) }{| z_{\text{B},k} |^2 + c_k}.
\end{equation}

\subsubsection{The Functional Derivative}
The variation in the constraint term of \eqref{eq:lagrangian} is
\begin{equation}
\delta \left[ \tau \left( Q - \int_{\mathcal{D}} \rho(\mathbf{x}) \mathrm{d}\mathbf{x} \right) \right] = - \tau \int_{\mathcal{D}} \zeta \eta(\mathbf{x}) \mathrm{d}\mathbf{x}.
\end{equation}

Thus, the total variation in $\mathcal{L}( \rho )$ is
\begin{equation}
\delta \mathcal{L} = \int_{\mathcal{D}} \left( G(\mathbf{x}) - \tau \right) \zeta \eta(\mathbf{x}) d\mathbf{x}.
\end{equation}

For $\delta \mathcal{L} = 0$ to hold for all admissible $\eta(\mathbf{x})$ and considering the box constraint \eqref{cons:density}, the optimal density function is
\begin{equation}\label{eq:opt_rho}
\rho_{\textrm{opt}}(\mathbf{x}) = 
\begin{cases}
\rho_{\max} & \text{if } G(\mathbf{x}) > \tau \\
\chi \in (0, \rho_{\max}) & \text{if } G(\mathbf{x}) = \tau \\
0 & \text{if } G(\mathbf{x}) < \tau
\end{cases}
\end{equation}
which reveals the density function is either maximal or zero, except for the case $G(\mathbf{x}) = \tau$.
The Lagrangian multiplier $\tau$ is updated via bisection to satisfy $\int_{\mathcal{D}} { \rho(\mathbf{x}) } d\mathbf{x} = Q$.
The solution \eqref{eq:opt_rho} resembles a spatial version of waterfilling, where $G(\mathbf{x})$ acts as the equivalent of channel gain in spatial domain. 
$\tau$ is the water level adjusted to consume total UAV numbers $Q$.
The implementation details are shown in Algorithm \ref{alg:dt-ara}, which is named as DT-ARA to emphasize density tuning in aerial networks.
The properties and advantages of the proposed DT-ARA algorithm are summarized in the following remarks.

\begin{algorithm}[!t]
	\caption{DT-ARA for ARIS Spatial Distribution Optimization}
	\label{alg:dt-ara}
	\begin{algorithmic}[1]
		\REQUIRE Net marginal gain function $G(\mathbf{x})$, total UAV quantity $Q$, maximum density $\rho_{\max}$, convergence tolerance for bisection.
		\ENSURE Optimized spatial density function $\rho_\text{opt}(\mathbf{x})$
		\STATE Set lower and upper bounds for the water level, $\tau_{\text{low}}$ and $\tau_{\text{high}}$.
		\WHILE{$(\tau_{\text{high}} - \tau_{\text{low}}) > \text{tolerance}$}
		\STATE Set $\tau_{\text{mid}} = (\tau_{\text{low}} + \tau_{\text{high}}) / 2$.
		\STATE For all $\mathbf{x} \in \mathcal{D}$, allocate density: 
		
		$\rho_{\text{temp}}(\mathbf{x}) = \begin{cases}
		\rho_{\max} & \text{if } G(\mathbf{x}) > \tau_{\text{mid}} \\
		0 & \text{otherwise}
		\end{cases}$
		\STATE Compute total allocated UAVs: $Q_{\text{temp}} = \int_{\mathcal{D}} \rho_{\text{temp}}(\mathbf{x}) d\mathbf{x}$
		\IF{ $Q_{\text{temp}} \leq Q$ }
		\STATE Set $\tau_{\text{high}} = \tau_{\text{mid}}$
		\ELSE
		\STATE Set $\tau_{\text{low}} = \tau_{\text{mid}}$
		\ENDIF
		\ENDWHILE
		\STATE \textbf{return} The optimized spatial density function $\rho_\text{opt}(\mathbf{x})$
	\end{algorithmic}
\end{algorithm}

\begin{remark}
The DT-ARA algorithm concentrates UAVs in regions where the net marginal gain function $G(\mathbf{x})$ is high. This aligns with the physical intuition that high $G(\mathbf{x})$ occurs where signal enhancement, i.e., $\mathrm{Re}(z_{\text{B},k}^* f_k(\mathbf{x}))$, dominates over potential interference amplification $P_\text{J} \gamma_k \mathrm{Re}(z_{\text{J},k}^* g_k(\mathbf{x}))$. 
On the other hand, low $G(\mathbf{x})$ regions either amplify interference or contribute weakly to signal strength.
\end{remark}

\begin{remark}
As derived in the Proposition \ref{prop_jammer_location}, optimal positioning balances proximity to users against directivity for interference mitigation. DT-ARA extends this to UAV deployment. The density is $\rho_{\max}$ in user-proximal zones with strong paths to BS/users and weak coupling to jammer, while the density becomes zero in interference-aligned zones facing the jammer or with poor signal paths.
\end{remark}

\begin{remark}
Consequently, the optimal deployment inherently avoids positioning ARIS nodes where they would amplify jamming signals. This spatial avoidance is a critical resilience feature, as it prevents the system from inadvertently strengthening interference in adversarial directions.
\end{remark}

The preceding variational derivation confirms that the spatial water-filling solution \eqref{eq:opt_rho} is mathematically optimal for UAV density optimization under the given constraints. 
Thus, the proposed DT-ARA algorithm provides an efficient means to this optimum by exploiting spatial heterogeneity in signal and interference channels, balancing the trade-off between signal enhancement and jamming risks, and converging to a resource allocation where UAVs are either maximally dense or absent.
This optimal strategy aligns with the observation that RIS-assisted anti-jamming requires strategic spatial avoidance of interference pathways while concentrating resources on high-gain links. 
More broadly, the framework extends the mean-field techniques to aerial networks, offering a foundational approach for robust deployment in contested environments.

\begin{algorithm}[t]
	\caption{Joint Anti-Jamming Optimization Framework}
	\label{alg:overall}
	\begin{algorithmic}[1]
		\REQUIRE Channel state information, estimated jammer position $\hat{\mathbf{j}}$, uncertainty radius $\epsilon$, BS power $P_\text{B}$, jamming power $P_\text{J}$, convergence threshold $\varepsilon$, max iterations $T_{\max}$.
		\ENSURE Optimized BS beamforming $\{\mathbf{w}_k\}$, RIS phases $\boldsymbol{\Theta}(\mathbf{x})$, UAV density $\rho(\mathbf{x})$.
		
		\STATE \textbf{Initialize:} $\rho^{(0)}(\mathbf{x})$, $\boldsymbol{\Theta}^{(0)}(\mathbf{x})$, set $t = 0$ 
		\STATE Compute worst-case jamming strategy $(\mathbf{j}, \mathbf{v})$ using Algorithm \ref{Jammer_opt_alg}.
		\REPEAT
		\STATE Update BS beamforming $\{\mathbf{w}_k\}^{(t)}$ via \eqref{eq:ZF+water_filling}.
		\STATE Optimize RIS phases $\boldsymbol{\Theta}^{(t)}(\mathbf{x})$ using Algorithm \ref{alg:ris-manifold-opt}.
		\STATE Optimize UAV density $\rho^{(t)}(\mathbf{x})$ using Algorithm \ref{alg:dt-ara}.
		\STATE Update $t$ with $t + 1$.
		\UNTIL $\left|R_{\text{sum}}^{(t)} - R_{\text{sum}}^{(t-1)}\right| < \varepsilon$ or $t > T_{\max}$
		\STATE \textbf{return} $\{\mathbf{w}_k\} = \{\mathbf{w}_k\}^{(t)}$, $\boldsymbol{\Theta}(\mathbf{x}) = \boldsymbol{\Theta}^{(t)}(\mathbf{x})$, $\rho(\mathbf{x}) = \rho^{(t)}(\mathbf{x})$.
	\end{algorithmic}
\end{algorithm}

\subsection{Overall Algorithm}
The proposed algorithm employs an alternating optimization framework to solve the joint design problem for anti-jamming communication systems. The procedure begins by determining the worst-case jamming strategy using Algorithm \ref{Jammer_opt_alg}, which computes the optimal jammer position and beamforming vector under uncertainty constraints. Subsequently, an iterative alternating optimization loop is executed, where the BS beamforming vectors, RIS phase configurations, and UAV spatial density distribution are sequentially updated while fixing other variables. This process iterates until the sum-rate improvement falls below a predefined convergence threshold. The implementation details are shown in Algorithm \ref{alg:overall}.


\begin{remark}\label{remark:jammer_strategy}
While Algorithm \ref{alg:overall} presents a joint optimization framework, its practical execution could leverage the different adaptation time scales of the system components.
Specifically, intelligent jammers operate on different timescales. For instance, their beamforming vector can be adjusted electronically with negligible latency, whereas repositioning suffers from significant delay due to mechanical constraints. 
This inherent hysteresis allows our system to operate a hierarchical defense. 
In particular, we can rapidly update the BS/RIS parameters to counteract jamming beam variations, while performing occasional recomputation of the ARIS density distribution, e.g., upon persistent performance degradation, to counter large-scale jammer repositioning. 
This approach efficiently balances computational tractability with robustness in dynamically contested environments.
\end{remark}

\section{Performance Analysis}\label{sec:performance}
This section analyzes the computational complexity and convergence properties of the proposed framework. We first break down the algorithmic complexity, emphasizing the scalability derived from the mean-field approach. Then, we establish the convergence guarantee of the alternating optimization procedure.
\subsection{Computational Complexity}
The overall computational complexity consists of four main components, with the total complexity given by
\begin{equation}
\mathcal{C}_{\textrm{total }}=\mathcal{C}_{\textrm{Jam }}+ T_{\textrm{alt}} (\mathcal{C}_{\textrm{BS}}+\mathcal{C}_{\textrm{RIS}}+\mathcal{C}_{\textrm{UAV}})
\end{equation}
where $\mathcal{C}_{\textrm{Jam }}$ denotes the computational complexity of Algorithm \ref{Jammer_opt_alg}; $\mathcal{C}_{\textrm{BS}}$ denotes the computational complexity of the BS beamforming; $\mathcal{C}_{\textrm{RIS}}$ denotes the complexity of Algorithm \ref{alg:ris-manifold-opt}; $\mathcal{C}_{\textrm{UAV}}$ denotes the complexity of Algorithm \ref{alg:dt-ara}; $T_{\textrm{alt}}$ denotes the number of alternating optimization iterations required for convergence.

Regarding $\mathcal{C}_{\textrm{Jam }}$, the dominant operations include eigenvalue decomposition of $\mathbf{R}(\mathbf{j})$, which costs $\mathcal{O}(L^3)$, and spatial gradient computation of $\nabla_{\mathbf{j}} \lambda_{\max}(\mathbf{R}(\mathbf{j}))$, which costs $\mathcal{O}(KL | \mathcal{G} |)$ and $| \mathcal{G} |$ is the number of grid points for feasible region $\mathcal{D}$ discretization.
\begin{equation}
\mathcal{C}_{\textrm{Jam }} = \mathcal{O}(L^3 + KL | \mathcal{G} |).
\end{equation}

Regarding $\mathcal{C}_{\textrm{BS}}$, the zero-forcing matrix inversion costs $\mathcal{O}(K^3)$ and water-filling power allocation costs $\mathcal{O}(K \log K)$, which yields
\begin{equation}
\mathcal{C}_{\textrm{BS}} = \mathcal{O}(K^3).
\end{equation}

Regarding $\mathcal{C}_{\textrm{RIS}}$, the functional gradient computation costs $\mathcal{O}(N^2 K)$ and the spatial integral evaluation costs $\mathcal{O}( | \mathcal{G} | N K)$. Thus we have 
\begin{equation}
\mathcal{C}_{\textrm{RIS}} = \mathcal{O}( T_{\textrm{RIS}} ( N^2 K + | \mathcal{G} | N K ) )
\end{equation}
where $T_{\textrm{RIS}}$ is the number of iterations of Algorithm \ref{alg:ris-manifold-opt}.

Regarding $\mathcal{C}_{\textrm{UAV}}$, evaluation gain function $G(\mathbf{x})$ costs $\mathcal{O}( | \mathcal{G} | K )$ and bisection search for $\lambda$ costs $\mathcal{O}(|\mathcal{G}| K \log(1 / \epsilon_{\lambda}))$. Thus we have
\begin{equation}
\mathcal{C}_{\textrm{UAV}} = \mathcal{O}(|\mathcal{G}| K \log(1 / \epsilon_{\lambda}))
\end{equation}
which scales linearly with spatial resolution $|\mathcal{G}|$.

\begin{remark}
A key advantage of the mean-field approach is that the computational complexity of the DT-ARA algorithm is independent of the number of UAVs $Q$. Instead, it scales linearly with the spatial discretization resolution $|\mathcal{G}|$. This makes the framework particularly efficient for managing large-scale ARIS swarms, as adding more drones does not increase the algorithmic computation cost.
\end{remark}

\subsection{Convergence Analysis}
Note that each iteration of the alternating optimization loop generates a non-decreasing sequence of sum-rate values. 
For example, considering the three sequential updates within iteration $t$, we have
\begin{align}
R_{\text{sum}}(\mathbf{w}^{(t)}, \boldsymbol{\Theta}^{(t)}, \rho^{(t)}) &\leq R_{\text{sum}}(\mathbf{w}^{(t + 1)}, \boldsymbol{\Theta}^{(t)}, \rho^{(t)}) \\
& \leq R_{\text{sum}}(\mathbf{w}^{(t+1)}, \boldsymbol{\Theta}^{(t + 1)}, \rho^{(t)}) \\
& \leq R_{\text{sum}}(\mathbf{w}^{(t + 1)}, \boldsymbol{\Theta}^{(t + 1)}, \rho^{(t + 1)}).
\end{align}

This monotonic improvement arises from the fact that each subproblem, e.g., optimizing the BS beamforming, the RIS phase shifts, or the UAV spatial density, is solved to either global optimality (as in the case of ZF beamforming with water-filling) or to a stationary point (as in the Riemannian manifold and variational optimization steps). Since each update either increases or maintains the sum-rate, the overall alternating optimization procedure is guaranteed to converge to a locally optimal solution. The boundedness of the objective function, due to the finite transmit power and interference constraints, ensures that the sequence of sum-rate converges to a finite value. Although the problem \eqref{eq:opt_1} is non-convex, the alternating approach efficiently decomposes the joint design into tractable subproblems, each of which contributes to the overall convergence behavior.

\section{Numerical Results}\label{sec:numerical_results}
In this section, we provide numerical results to verify our analysis and evaluate the performance of the proposed algorithms. 
The performance is evaluated under various system parameters and compared against multiple benchmark schemes.
\subsection{Simulation Setup}
\begin{table}[!t]
	\centering
	\caption{System Parameters for Numerical Simulation}
	\label{tab:system_params}
	\begin{tabular}{ll}
		\toprule
		\textbf{Parameter} & \textbf{Value} \\ 
		\midrule
		BS antennas, $M$ & $16$  \\
		RIS elements per UAV, $N$ & $20$-$100$  \\
		UAV count, $Q$ & $100$-$1000$ \\
		Users count, $K$ & $4, 8$ \\
		Jammer antennas, $L$ & $16$ \\
		UAV deployment region, $\mathcal{D}$ & $200 \times 200 \text{ m}^2$ \\
		Number of grid points for $\mathcal{D}$, $|\mathcal{G}|$ & $20 \times 20$ \\
		UAV altitude & $100$ m \cite{Wang_D2025}\\
		Path-loss exponent, $\alpha$ & $2.2$ \cite{Hou_Z2023}  \\
		Reference channel gain, $\beta$ & $-30$ dB \\
		Rician K-factor, $\kappa$ & $10$ dB \cite{Wang_D2025} \\
		BS maximum transmit power, $P_\text{B}$ & $30$-$70$ dBm \\
		Maximum jamming power, $P_\text{J}$ & $40$-$80$ dBm \\
		Jamming position uncertainty, $\epsilon$ & $0$-$30$ m \\
		Maximum density, $\rho_{\max}$ & $0.02$-$0.2$ UAVs/m$^2$ \\
		Noise power, $\sigma_k^2$ & $-102$ dBm \cite{Peng_H2023}\\
		\bottomrule
	\end{tabular}
\end{table}

\begin{figure*}[t] 
	\centering
	\includegraphics[width=0.65\textwidth]{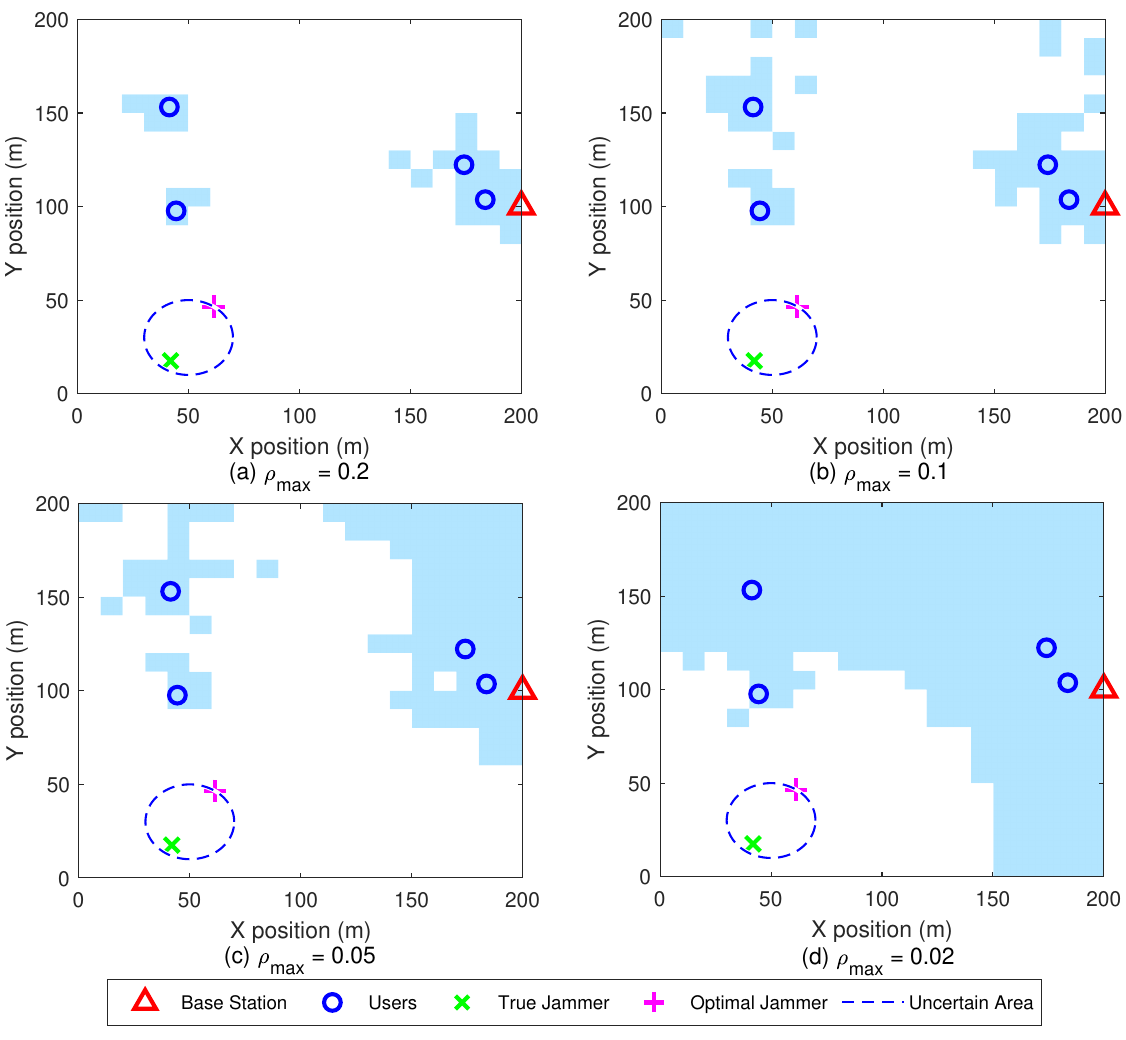} 
	\caption{\raggedright ARIS density distribution under different $\rho_{\max}$, where each grid denotes a $10$ m $\times$ $10$ m area. The blue grids indicate maximum ARIS density and white grids denote no ARIS deployment.} 
	\label{fig:UAV_distribution}
\end{figure*}

The main parameters of the simulation are list in Table \ref{tab:system_params}.
The BS and users are fixed at ground level, and the jammer operates within uncertainty region $\mathcal{J}$.
Given the characteristics of air-ground propagation environment, the channel model contains LoS paths due to elevated drone deployments, and non-line-of-sight (NLoS) components arising from terrestrial blockages and scattering.
Take the BS-ARIS channel $\mathbf{H}_{\text{B,U}}$ as an example, which combines LoS and NLoS contributions as \cite{Wang_D2025}
\begin{equation}
\mathbf{H}_{\text{B,U}} = \sqrt{\beta d_{\text{B,U}}^{-\alpha}} \bigg( \sqrt{\frac{\kappa}{\kappa + 1}} \mathbf{H}_{\text{LoS}} + \sqrt{\frac{1}{\kappa + 1}}  \mathbf{H}_{\text{NLoS}} \bigg)
\end{equation}
where $\beta$ is the channel power gain measured at 1 m; $\alpha$ is the path loss exponent; $d_{\text{B,U}}$ denotes the Euclidean distance between BS and ARIS; $\kappa$ is the Rician K-factor. Assuming the antenna elements are arranged in uniform linear array, the LoS component is
\begin{equation}
\mathbf{H}_{\text{LoS}}(\mathbf{x}) = \tilde{\mathbf{a}}_\text{r}(\theta) \tilde{\mathbf{a}}_\text{t}^\dagger(\phi) 
\end{equation}
where $\tilde{\mathbf{a}}_\text{r}(\theta) \in \mathbb{C}^{N}$ and $\tilde{\mathbf{a}}_\text{t}(\phi) \in \mathbb{C}^{M}$ are array response vectors with angle of departure $\theta$ and 
angle of arrival $\phi$. The expressions are omitted for brevity.
$\mathbf{H}_{\text{NLoS}} \in \mathbb{C}^{N \times M}$ models diffuse multipath via independent and identically distributed entries following $\mathcal{CN}(0 , 1)$.
Similarly, the ARIS-to-user channel and jammer-to-ARIS channel follow analogous formulations.

We compare the performances of the following $7$ schemes.
\begin{enumerate}[topsep=0pt]
\item \textit{Proposed Robust Joint Optimization:} This is our proposed scheme based on Algorithm \ref{alg:overall}.
\item \textit{Scheme 1:} This is our proposed scheme without considering jammer's location uncertainty. The estimated jammer location is treated as the true site, and thus Scheme 1 is a non-robust version of the proposed scheme.
\item \textit{Scheme 2:}  This is our proposed scheme with ARIS uniformly distributed in the feasible region and ARIS phase shifts optimized using Algorithm \ref{alg:ris-manifold-opt}.
\item \textit{Scheme 3:}  This is our proposed scheme with ARIS uniformly distributed in the feasible region and ARIS phase shifts randomly generated.
\item \textit{Scheme 4 \cite{Liu_J2024}:} Deploying a single dynamic ARIS to serve multiple users under the interference of an adjacent cell, with joint optimization of BS beamforming, ARIS phase shift, and ARIS location.
\item \textit{Scheme 5 \cite{Shang_B2023}:}  Deploying multiple ARIS to serve multiple users with no countermeasure to the jammer. The sum-rate is maximized by joint optimization of BS transmit beamforming, ARIS phase shift, and ARIS placement. For a fair comparison, the number of ARIS is identical to that of the proposed scheme.
\item \textit{Scheme 6 \cite{Tang_X2021}:} Deploying a single static ARIS to serve a single user in presence of an intentional jammer, with joint optimization of the ARIS position and phase shift.	
\end{enumerate}

\subsection{Results Analysis}

Fig. \ref{fig:UAV_distribution} presents the ARIS density distribution under varying maximum density constraints $\rho_{\max} \in \{ 0.2, 0.1, 0.05, 0.02 \}$ in the operational area with $10$ m grid resolution, featuring a BS, $4$ randomly distributed users, and a jammer with true position and derived optimal jamming position with uncertainty radius $\epsilon = 30$ m. The BS transmit power is $40$ dBm and the jamming power is $50$ dBm. 
Three observations can be drawn from the results. 
First, the density distribution reveals that the ARISs concentrate near BS and users while avoiding jammer-proximal regions, aligning with the proposed DT-ARA algorithm's principle where the net marginal gain function $G(\mathbf{x})$ prioritizes positions with high SINR. 
Through binary thresholding, it activates $\rho(\mathbf{x}) = \rho_{\max}$ only where $G(\mathbf{x}) > \tau$ and suppresses deployment near interference sources since jamming power dominates the denominator of $G(\mathbf{x})$ in jammer-affected zones.
Second, an inverse relationship emerges between maximum density $\rho_{\max}$ and spatial coverage.
Higher $\rho_{\max}$, e.g., $\rho_{\max} = 0.2$, forces the ARISs into sparse high-density clusters at critical signal pathways near BS/users to satisfy the constraint $\rho(\mathbf{x}) \leq \rho_{\max}$ while maintaining total UAV count $Q$.
Lower $\rho_{\max}$, e.g., $\rho_{\max} = 0.02$, distributes ARISs across broader areas as the proposed algorithm must activate more locations to accommodate $Q$, sacrificing beamforming intensity for spatial diversity.
Third, the optimal jammer position locates at the uncertainty region boundary towards the user cluster centroid, as predicted by Proposition \ref{prop_jammer_location}, which prioritizes distance minimization to dominant user groups when they are clustered.

\begin{figure}[t] 
	\centering
	\includegraphics[width=0.4\textwidth]{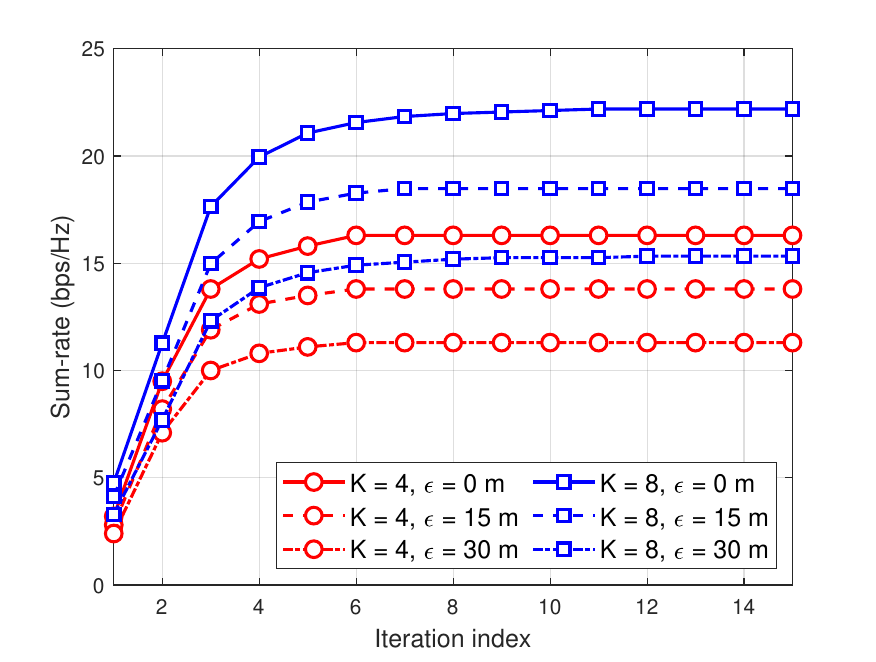} 
	\caption{\raggedright Convergence of the proposed algorithm under different users count $K$ and jammer position uncertainties.} 
	\label{fig:convergence}
\end{figure}

Fig. \ref{fig:convergence} illustrates the convergence behavior of the proposed robust joint optimization algorithm under different system configurations.
The simulation considers two user scenarios, i.e., $K = 4$ and $K = 8$, and three jammer position uncertainty radii.
The jamming power is $P_{\textrm{J}} = 50$ dBm. The BS transmit power for $K = 4$ is $P_{\textrm{B}} = 40$ dBm and for $K = 8$ is $P_{\textrm{B}} = 50$ dBm.
The RIS elements per UAV is $N = 50$ and total UAV count is $Q = 100$ with $\rho_{\max} = 0.05$.
Key observations reveal that all configurations achieve convergence within $8$ iterations, confirming the proposed algorithm's computational efficiency.
The $K = 4$ configuration stabilizes at approximately $6$ iterations, while $K = 8$ requires approximately $8$ iterations due to increased coordination complexity. 
Jamming position uncertainty degrades performance across both user configurations, where larger error radii results in lower sum-rates compared to perfect jammer location knowledge.
Moreover, the $8$-user configuration exhibits greater sensitivity to jammer position uncertainty compared to the $4$-user case, as the increased number of users disperses the beamforming gain and complicates interference alignment under larger $\epsilon$.
Overall, the algorithm demonstrates stable convergence across all tested scenarios, validating its robustness for practical implementation in contested environments with uncertain jammer positions.

Fig. \ref{fig:sumrate_vs_pj} shows the sum-rate performances of various ARIS-assisted anti-jamming schemes against increasing jamming power from $40$ dBm to $80$ dBm. 
The BS serves $4$ users and transmit power is $45$ dBm.
The RIS elements per UAV is $N = 100$ with total UAV count $Q = 100$ with $\rho_{\max} = 0.05$.
The proposed robust joint optimization scheme demonstrates superior anti-jamming performance, maintaining the highest sum-rates across all jamming power levels due to its integrated optimization of UAV spatial distribution, RIS phase configurations, and robust beamforming that explicitly accounts for jammer position uncertainty.
Specifically, at $40$ dBm jamming power, the proposed scheme achieves $26$ bps/Hz, outperforming Scheme 2 (uniform UAV distribution with optimized phases) by $44$\%, Scheme 3 (uniform UAV distribution with random phases) by $150$\%, and state-of-the-art reference schemes by $108$-$300$\%.
As jamming power increases to $80$ dBm, the proposed scheme's sum-rate declines to $10$ bps/Hz, which is still acceptable in such highly unfavorable scenario.

Comparative analysis reveals that Scheme 2 maintains reasonable performance at moderate jamming levels but suffers performance loss due to unoptimized spatial resource allocation, while Scheme 3 shows higher vulnerability.
The reference schemes demonstrate fundamental limitations.
Specifically, the single dynamic ARIS approach in Scheme 4 lacks spatial diversity.
The multi-ARIS solution without jamming consideration in Scheme 5 shows inadequate interference suppression, and the static ARIS deployment in Scheme 6 exhibits the worst performance due to the weakness of a single fixed UAV. 
These results collectively validate that joint optimization of spatial distribution and RIS reflecting properties through the proposed framework is essential for maintaining reliable communications in contested environments with adaptive jammers.

\begin{figure}[t] 
	\centering
	\includegraphics[width=0.4\textwidth]{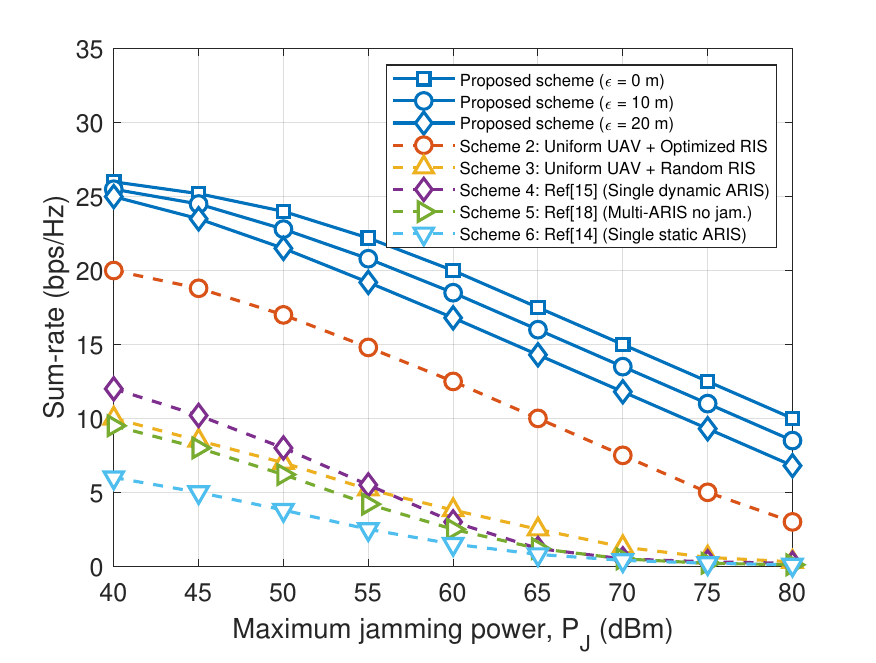} 
	\caption{\raggedright Sum-rate against jamming power under different anti-jamming schemes.} 
	\label{fig:sumrate_vs_pj}
\end{figure}

\begin{figure}[t] 
	\centering
	\includegraphics[width=0.4\textwidth]{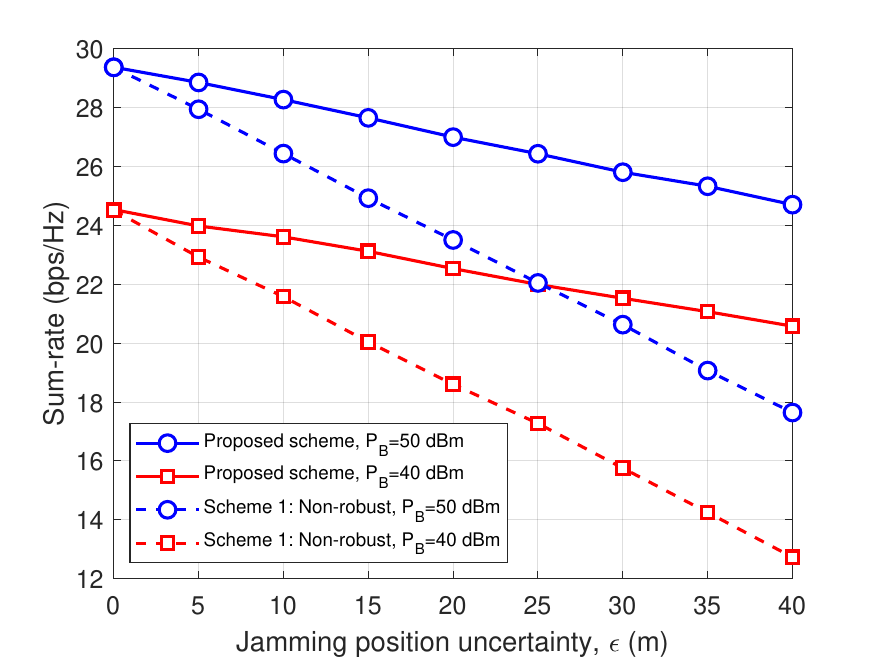} 
	\caption{\raggedright Sum-rate against jamming position error radius.} 
	\label{fig:sumrate_vs_epsilon}
\end{figure}

Fig. \ref{fig:sumrate_vs_epsilon} illustrates the impact of jammer position uncertainty on the sum-rate performance of the proposed robust joint optimization framework and the non-robust version (Scheme 1), which ignores jammer position uncertainty. 
The results demonstrate that the proposed robust scheme consistently outperforms Scheme 1 across all error radii, with the performance gap widening as uncertainty increases.
At $50$ dBm transmit power, the proposed scheme's sum-rate declines from $29.3$ to $24.7$ bps/Hz, while Scheme 1 suffers a more severe degradation from $29.3$ to $17.6$ bps/Hz.
Similarly, at $40$ dBm transmit power, the proposed scheme decreases from $24.5$ to $20.6$ bps/Hz compared to Scheme 1's drop from $24.5$ to $12.7$ bps/Hz.
The robustness of the proposed framework is especially valuable at higher uncertainty levels, where non-robust designs degrades dramatically. 
These results highlight the critical importance of incorporating jammer position uncertainty into the system design, particularly in contested environments where inaccurate jammer location estimation would otherwise lead to significant performance degradation.

\begin{figure}[t] 
	\centering
	\includegraphics[width=0.4\textwidth]{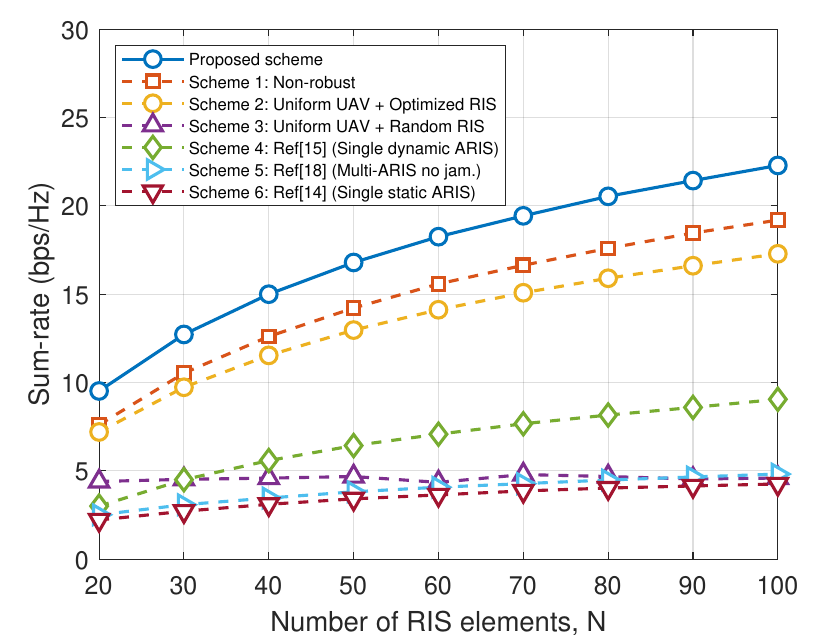} 
	\caption{\raggedright Sum-rate versus number of RIS elements under different anti-jamming schemes.} 
	\label{fig:sumrate_vs_N}
\end{figure}

Fig. \ref{fig:sumrate_vs_N} illustrates the sum-rate performance versus the number of RIS elements for different anti-jamming transmission schemes under a fixed jammer position error of $\epsilon = 15$ m. 
The simulation results demonstrate that the proposed robust joint optimization framework consistently outperforms all other schemes across the entire range of RIS elements, achieving sum-rates from approximately $9.5$ to $22.3$ bps/Hz.
Increasing the number of RIS elements enhances the system's anti-jamming capability by providing higher spatial degrees of freedom for precise beamforming towards legitimate users and creating deeper nulls towards jammers. 
Scheme 1 shows competitive performance but suffers from significant degradation due to its inability to account for jammer uncertainty, while Scheme 2 performs moderately well but fails to match the proposed scheme's performance due to the lack of spatial resource allocation. 
The remaining schemes exhibit substantially lower performance, with Scheme 3 showing minimal improvement with increasing RIS elements, and the reference schemes 4-6 demonstrating limitations in jamming mitigation capability.

\begin{figure}[t] 
	\centering
	\includegraphics[width=0.4\textwidth]{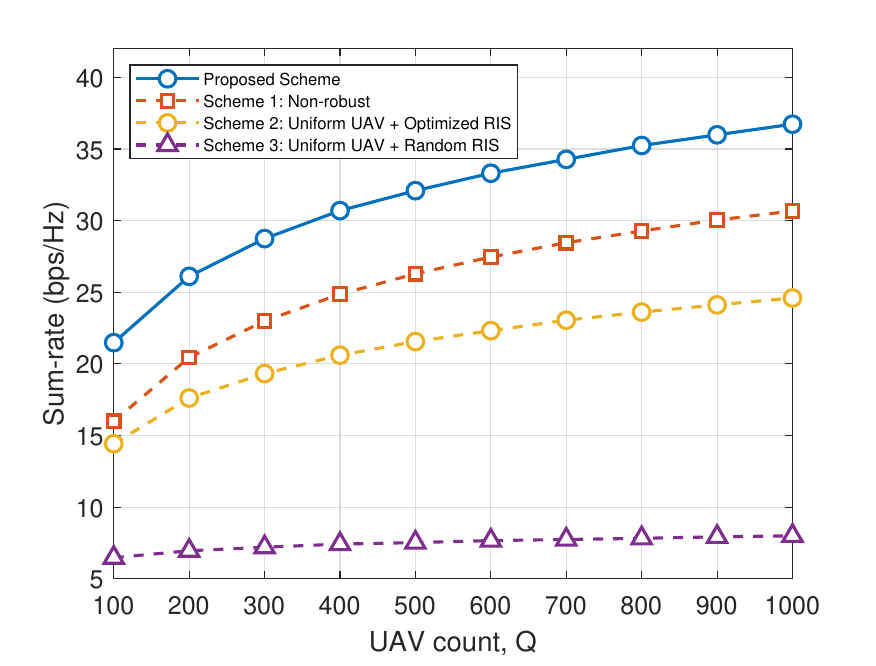} 
	\caption{\raggedright Sum-rate versus number of UAVs under different configurations.} 
	\label{fig:sumrate_vs_Q}
\end{figure}

Fig.\ref{fig:sumrate_vs_Q} illustrates the sum-rate performance versus the number of UAVs for four different anti-jamming transmission schemes under a fixed jammer position error of $\epsilon = 20$ m. 
The simulation results demonstrate that the proposed robust joint optimization framework consistently outperforms all other schemes across the entire range of UAV counts, achieving sum-rates from approximately $21.5$ to $36.7$ bps/Hz.
Scheme 3 shows minimal performance improvement with additional UAVs, emphasizing the critical importance of both optimized phase configuration and strategic spatial deployment.
Increasing the number of UAVs enhances the system's anti-jamming resilience primarily through spatial macro diversity and distributed jamming mitigation. 
A larger swarm of ARISs provides greater spatial degrees of freedom, enabling more flexible and robust coverage reconfiguration. 
This allows the system to strategically position reflecting nodes to create stronger signal pathways to legitimate users while simultaneously avoiding or nullifying jamming directions. 
The proposed DT-ARA algorithm provides significant performance gains in optimizing the UAV spatial distribution, effectively forming a distributed and agile electromagnetic barrier that counters the malicious jammer.

\section{Conclusion}\label{sec:conclusion}
This paper proposed a mean-field-based framework for robust anti-jamming transmission in large-scale ARIS-assisted networks, where UAV spatial density, RIS phase shifts, and BS beamforming are jointly optimized under jammer position uncertainty. The results reveal two important insights. First, optimal ARIS deployment is non-uniform, since the proposed DT-ARA algorithm concentrates drones in high-gain regions while avoiding interference-prone areas. Second, the worst-case optimization formulation is crucial, as it guarantees resilience against an adaptive jammer capable of optimizing both its beamforming and position. Overall, the joint optimization framework achieves significant sum-rate gains over benchmarks, validating the framework's effectiveness for ARIS-assisted anti-jamming communications.

\appendix
\subsection{Proof of Proposition 1}\label{appendix1}
We first consider the interior critical point condition.
Assume that an interior point $\mathbf{j}' \in \mathrm{int}(\mathcal{J})$ satisfies $\nabla_{\mathbf{j}} \lambda_{\max}(\mathbf{R}(\mathbf{j}')) = 0$.
The Karush-Kuhn-Tucker (KKT) conditions for the constrained optimization problem \eqref{jammer_location_prob} require the Lagrangian stationarity condition
\begin{equation}
\nabla_{\mathbf{j}} \lambda_{\max}(\mathbf{R}(\mathbf{j}')) - 2\mu (\mathbf{j}' - \hat{\mathbf{j}}) = \mathbf{0}.
\end{equation}
where $\mu$ is the Lagrangian multiplier.
When $\nabla_{\mathbf{j}} \lambda_{\max}(\mathbf{R}(\mathbf{j}')) = \mathbf{0}$, we have $\mu (\mathbf{j}' - \hat{\mathbf{j}}) = \mathbf{0}$, which implies that either $\mu = 0$, $\mathbf{j}'$ is an unconstrained local maximum; or $\mathbf{j}' = \hat{\mathbf{j}}$, $\hat{\mathbf{j}}$ is a stationary point. The latter case is generally suboptimal, so only unconstrained maxima satisfy the condition.

When no interior critical point exists, the stationarity condition becomes
\begin{equation}
\nabla_{\mathbf{j}} \lambda_{\max}(\mathbf{R}(\mathbf{j})) = 2\mu (\mathbf{j} - \hat{\mathbf{j}})
\end{equation}
which yields $\mathbf{j} - \hat{\mathbf{j}} = \frac{1}{2\mu} \nabla_{\mathbf{j}} \lambda_{\max}(\mathbf{R}(\mathbf{j}))$. 
Enforcing the norm constraint 
\begin{equation}
\left\| \frac{1}{2\mu} \nabla_{\mathbf{j}} \lambda_{\max}(\mathbf{R}(\mathbf{j})) \right\| = \epsilon
\end{equation}
yields $\frac{1}{2\mu} = \frac{\epsilon}{\| \nabla_{\mathbf{j}} \lambda_{\max}(\mathbf{R}(\mathbf{j})) \|}$. Thus, the optimal position is
\begin{equation}
\mathbf{j}_\text{opt} = \hat{\mathbf{j}} + \epsilon \frac{\nabla_{\mathbf{j}} \lambda_{\max}(\mathbf{R}(\mathbf{j}))}{\| \nabla_{\mathbf{j}} \lambda_{\max}(\mathbf{R}(\mathbf{j})) \|}
\end{equation}
which completes the proof.

\subsection{Proof of Proposition 2} \label{appendix2}
With the ZF beamforming, the sum-rate $R_{\textrm{sum}}$ and SINR $\gamma_k$ are given by
\begin{align}
R_{\textrm{sum}} &= \frac{1}{\ln 2} \sum_{k = 1}^{K} \ln(1 + \gamma_k) \\
\gamma_{k} &= \frac{\left|z_{\text{B}, k}\right|^{2}}{c_{k}}
\end{align}
where $z_{\text{B}, k}$ and $c_k$ are given in \eqref{eq:z_Bk} and \eqref{eq:c_k}, respectively.
The aggregate channels defined via spatial integrals are 
\begin{align}
z_{\text{B}, k} &= \int_{\mathcal{D}} \underbrace{\mathbf{h}_{\text{U}, k}^\dagger\left(\mathbf{x}\right) \boldsymbol{\Theta}\left(\mathbf{x}\right) \mathbf{H}_{\text{B, U}}\left(\mathbf{x}\right)}_{\tilde{g}_{\text{B}, k}\left(\mathbf{x}\right)} \mathbf{w}_{k} \rho(\mathbf{x}) \mathrm{d} \mathbf{x} \\
z_{\text{J}, k} &= \int_{\mathcal{D}} \underbrace{\mathbf{h}_{\text{U},k}^\dagger(\mathbf{x}) \boldsymbol{\Theta}(\mathbf{x}) \mathbf{H}_{\text{J,U}}(\mathbf{x},\mathbf{j})}_{\tilde{g}_{\text{J}, k}\left(\mathbf{x}\right)} \mathbf{v} \rho(\mathbf{x}) \mathrm{d}\mathbf{x} 
\end{align}
where $\tilde{g}_{\text{B}, k}\left(\mathbf{x}\right)$ and $\tilde{g}_{\text{J}, k}\left(\mathbf{x}\right)$ are position-dependent signal and interference responses, expressed by
\begin{align}
\tilde{g}_{\text{B}, k}\left(\mathbf{x}\right) &= \sum_{n = 1}^{N} \underbrace{\left[\mathbf{h}_{\text{U}, k}\left(\mathbf{x}\right)\right]_{n}^{*}\left[\mathbf{H}_{\text{B, U}}\left(\mathbf{x}\right) \mathbf{w}_{k}\right]_{n}}_{c_{\text{B}, k, n}\left(\mathbf{x}\right)} \theta_{n}\left(\mathbf{x}\right) \\
\tilde{g}_{\text{J}, k}\left(\mathbf{x}\right) &= \sum_{n = 1}^{N} \underbrace{\left[\mathbf{h}_{\text{U}, k}\left(\mathbf{x}\right)\right]_{n}^{*}\left[\mathbf{H}_{\text{J, U}}\left(\mathbf{x}, j\right) \mathbf{v}\right]_{n}}_{c_{\text{J}, k, n}\left(\mathbf{x}\right)} \theta_{n}\left(\mathbf{x}\right).
\end{align}

The functional derivative of $R_{\textrm{sum}}$ w.r.t. $\theta_n(\mathbf{x})$ is 
\begin{equation}
\frac{\delta R_{\mathrm{sum}}}{\delta \theta_{n}(\mathbf{x})}=\frac{1}{\ln 2} \sum_{k = 1}^{K} \frac{1}{1 + \gamma_{k}} \frac{\delta \gamma_{k}}{\delta \theta_{n}(\mathbf{x})}
\end{equation}

From the expression of $\gamma_k$, we compute its derivative using the quotient rule
\begin{equation}\label{eq:delta_gamma_k}
\frac{\delta \gamma_{k}}{\delta \theta_{n}(\mathbf{x})}=\frac{1}{c_{k}^{2}}\left(c_{k} \frac{\delta\left|z_{\text{B}, k}\right|^{2}}{\delta \theta_{n}(\mathbf{x})}-\left|z_{\text{B}, k}\right|^{2} \frac{\delta c_{k}}{\delta \theta_{n}(\mathbf{x})}\right)
\end{equation}
where
\begin{equation}\label{eq:z_bk}
\frac{\delta\left|z_{\text{B}, k}\right|^{2}}{\delta \theta_{n}(\mathbf{x})} = 2\mathrm{Re}\left(z_{\text{B}, k}^{*} \frac{\delta z_{\text{B}, k}}{\delta \theta_{n}(\mathbf{x})}\right)
\end{equation}
and the functional derivative of $z_{\text{B}, k}$ w.r.t. $\theta_{n}(\mathbf{x})$ is
\begin{align}
\frac{\delta z_{\text{B}, k}}{\delta \theta_{n}(\mathbf{x})} &= \frac{\delta}{\delta \theta_{n}(\mathbf{x})} \int_{\mathcal{D}} \tilde{g}_{\text{B}, k}\left(\mathbf{x}\right) \rho\left(\mathbf{x}\right) \mathrm{d} \mathbf{x} \\
& =c_{\text{B}, k, n}(\mathbf{x}) \rho(\mathbf{x}).
\end{align}

Similarly, the functional derivative of $z_{\text{J}, k}$ is 
\begin{equation}
\frac{\delta z_{\text{J}, k}}{\delta \theta_{n}(\mathbf{x})} = c_{\text{J}, k, n}(\mathbf{x})\rho(\mathbf{x})
\end{equation}
and thus
\begin{equation}\label{eq:c_k_1}
\frac{\delta c_{k}}{\delta \theta_{n}(\mathbf{x})} = 2P_\text{J}\mathrm{Re}\left(z_{\text{J},k}^* c_{\text{J},k,n}(\mathbf{x})\rho(\mathbf{x})\right).
\end{equation}

Substituting \eqref{eq:z_bk}-\eqref{eq:c_k_1} into \eqref{eq:delta_gamma_k} yields
\begin{align}
\frac{\delta \gamma_{k}}{\delta \theta_{n}(\mathbf{x})}= &\frac{2\rho(\mathbf{x})}{c_{k}^{2}} (c_{k} \mathrm{Re}\left(z_{\text{B}, k}^{*} c_{\text{B}, k, n}(\mathbf{x})\right) \notag\\
&-\left|z_{\text{B}, k}\right|^{2} P_\text{J} \mathrm{Re}\left(z_{\text{J}, k}^{*} c_{\text{J}, k, n}(\mathbf{x})\right) ).
\end{align}

Using the definitions of $B_k$ and $C_k$ completes the proof.


\begin{thebibliography}{1}
	
\bibitem{Pirayesh_H2022} H. Pirayesh and H. Zeng, ``Jamming attacks and anti-jamming strategies in wireless networks: a comprehensive survey,'' \textit{IEEE Commun. Surveys Tuts.}, vol. 24, no. 2, pp. 767-809, 2nd Quart., 2022.

\bibitem{Darsena_D2022} D. Darsena and F. Verde, ``Anti-jamming beam alignment in millimeter-wave MIMO systems,'' \textit{IEEE Trans. Commun.}, vol. 70, no. 8, pp. 5417-5433, Aug. 2022.

\bibitem{Renzo_DM2020} M. Di Renzo et al., ``Smart radio environments empowered by reconfigurable intelligent surfaces: how it works, state of research, and the road ahead,'' \textit{IEEE J. Sel. Areas Commun.}, vol. 38, no. 11, pp. 2450-2525, Nov. 2020.

\bibitem{Wu_Q2021} Q. Wu, S. Zhang, B. Zheng, C. You and R. Zhang, ``Intelligent reflecting surface-aided wireless communications: A tutorial,'' \textit{IEEE Trans. Commun.}, vol. 69, no. 5, pp. 3313-3351, May 2021.

\bibitem{Huang_C2020} C. Huang et al., ``Holographic MIMO surfaces for 6G wireless networks: opportunities, challenges, and trends, \textit{IEEE Wireless Commun.}, vol. 27, no. 5, pp. 118-125, Oct. 2020.

\bibitem{Mahmoud_A2021} A. Mahmoud, S. Muhaidat, P. C. Sofotasios, et al., ``Intelligent reflecting surfaces assisted UAV communications for IoT networks: performance analysis,'' \textit{IEEE Trans. Green Commun. Netw.}, vol. 5, no. 3, pp. 1029-1040, Sept. 2021.

\bibitem{Yang_H2021} H. Yang et al., ``Intelligent reflecting surface assisted anti-jamming communications: a fast reinforcement learning approach,'' \textit{IEEE Trans. Wireless Commun.}, vol. 20, no. 3, pp. 1963-1974, Mar. 2021.

\bibitem{Zhang_M2024} M. Zhang, Z. Chu, Z. Zhu, D. Mi, L. Zhen and G. Chen, ``RIS-assisted anti-jamming wireless powered communication network,'' \textit{IEEE Wireless Commun. Lett.}, vol. 13, no. 10, pp. 2732-2736, Oct. 2024.

\bibitem{Dong_H2024} H. Dong, C. Hua, L. Liu, W. Xu and S. Guo, ``Optimization-driven DRL-based joint beamformer design for IRS-aided ITSN against smart jamming attacks,'' \textit{IEEE Trans. Wireless Commun.}, vol. 23, no. 1, pp. 667-682, Jan. 2024.

\bibitem{ElMossallamy_2021} M. A. ElMossallamy, K. G. Seddik, W. Chen, et al., ``RIS optimization on the complex circle manifold for interference mitigation in interference channels,'' \textit{IEEE Trans. Veh. Technol.}, vol. 70, no. 6, pp. 6184-6189, June 2021.

\bibitem{Sun_Y2024} Y. Sun et al., ``Active-passive cascaded RIS-aided receiver design for jamming nulling and signal enhancing,'' \textit{IEEE Trans. Wireless Commun.}, vol. 23, no. 6, pp. 5345-5362, June 2024.

\bibitem{Zou_C2023} C. Zou et al., ``Multi-layer RIS-assisted anti-jamming communications: a hierarchical game learning approach, \textit{IEEE Commun. Lett.}, vol. 27, no. 11, pp. 2998-3002, Nov. 2023.

\bibitem{Zhou_T2024} T. Zhou, K. Xu, G. Hu, X. Xia, W. Xie and C. Li, ``Robust beamforming design for STAR-RIS-assisted anti-jamming and secure transmission,'' \textit{IEEE Trans. Green Commun. Netw.}, vol. 8, no. 1, pp. 345-361, Mar. 2024.

\bibitem{Tang_X2021} X. Tang, D. Wang, R. Zhang, Z. Chu and Z. Han, ``Jamming mitigation via aerial reconfigurable intelligent surface: passive beamforming and deployment optimization,'' \textit{IEEE Trans. Veh. Technol.}, vol. 70, no. 6, pp. 6232-6237, June 2021.

\bibitem{Liu_J2024} J. Liu and H. Zhang, ``Dynamic aerial reconfigurable intelligent surface aided multi-cell multi-user communications,'' \textit{IEEE Trans. Wireless Commun.}, vol. 23, no. 11, pp. 16453-16465, Nov. 2024.

\bibitem{Yang_H2024_1} H. Yang, K. Lin, L. Xiao, Y. Zhao, Z. Xiong and Z. Han, ``Energy harvesting UAV-RIS-assisted maritime communications based on deep reinforcement learning against jamming,'' \textit{IEEE Trans. Wireless Commun.}, vol. 23, no. 8, pp. 9854-9868, Aug. 2024.

\bibitem{Hou_Z2023} Z. Hou et al., ``Joint IRS selection and passive beamforming in multiple IRS-UAV-enhanced anti-jamming D2D communication networks,'' \textit{IEEE Internet Things J.}, vol. 10, no. 22, pp. 19558-19569, Nov. 2023.

\bibitem{Shang_B2023} B. Shang, E. S. Bentley and L. Liu, ``UAV swarm-enabled aerial reconfigurable intelligent surface: modeling, analysis, and optimization,'' \textit{IEEE Trans. Commun.}, vol. 71, no. 6, pp. 3621-3636, June 2023.



\bibitem{Wang_D2024} D. Wang et al., ``Mean field game-based waveform precoding design for mobile crowd integrated sensing, communication, and computation systems,'' \textit{IEEE Trans. Wireless Commun.}, vol. 23, no. 8, pp. 10430–10444, Aug. 2024.

\bibitem{Jiang_J2025} J. Jiang, W. Chen, X. Guo, et al., ``A tractable approach to massive communication and ubiquitous connectivity in 6G standardization,'' \textit{ArXiv}, 2025.

\bibitem{Benamor_A2022} A. Benamor, O. Habachi, I. Kammoun and J. -P. Cances, ``Mean field game-theoretic framework for distributed power control in hybrid NOMA,'' \textit{IEEE Trans. Wireless Commun.}, vol. 21, no. 12, pp. 10502-10514, Dec. 2022.

\bibitem{Gao_H2022} H. Gao et al., ``Energy-efficient velocity control for massive numbers of UAVs: a mean field game approach,'' \textit{IEEE Trans. Veh. Technol.}, vol. 71, no. 6, pp. 6266-6278, June 2022.

	

\bibitem{Zhang_LH2020} H. Zhang, B. Di, L. Song, and Z. Han, ``Reconfigurable intelligent surfaces assisted communications with limited phase shifts: How many phase shifts are enough?'' \textit{IEEE Trans. Veh. Technol.}, vol. 69, no. 4, pp. 4498–4502, Apr. 2020.	

\bibitem{Zhang_S2022} S. Zhang, H. Zhang, B. Di, et al., ``Intelligent omni-surfaces: ubiquitous wireless transmission by reflective-refractive metasurfaces,'' \textit{IEEE Trans. Wireless Commun.}, vol. 21, no. 1, pp. 219-233, Jan. 2022.

\bibitem{Emil_2014} E. Bj\"{o}rnson, M. Bengtsson, and B. Ottersten, ``Optimal multiuser transmit beamforming: A difficult problem with a simple solution structure [Lecture Notes],'' \textit{IEEE Signal Process. Mag.}, vol. 31, no. 4, pp. 142–148, Jul. 2014.

\bibitem{Brunt_2004} B. van Brunt, \textit{The Calculus of Variations.} Springer, New York, 2004.



\bibitem{Wang_D2025} D. Wang et al., ``Secure energy efficiency for ARIS networks with deep learning: active beamforming and position optimization,'' \textit{IEEE Trans. Wireless Commun.}, vol. 24, no. 6, pp. 5282-5296, Jun. 2025.

\bibitem{Peng_H2023} H. Peng and L. -C. Wang, ``Energy harvesting reconfigurable intelligent surface for UAV based on robust deep reinforcement learning,'' \textit{IEEE Trans. Wireless Commun.}, vol. 22, no. 10, pp. 6826-6838, Oct. 2023.






	
	
	
	
	
\end{thebibliography}
\end{document}